\documentclass[11pt]{article}

\usepackage[margin=1in]{geometry}
\usepackage{amsmath,amssymb,amsthm,mathtools}
\usepackage{microtype}
\usepackage{booktabs}
\usepackage{xcolor}
\usepackage{enumitem}
\usepackage[breakable]{tcolorbox}

\newtcolorbox{parameterbox}{
  breakable,
  colback=white,
  colframe=black,
  colbacktitle=black!5,
  coltitle=black,
  fonttitle=\bfseries,
  title={Choice of parameters.},
  title after break={Choice of parameters (continued).},
  boxrule=0.5pt,
  arc=0pt,
  outer arc=0pt,
  boxsep=0pt,
  left=8pt,
  right=8pt,
  top=7pt,
  bottom=7pt,
  toptitle=5pt,
  bottomtitle=5pt,
  before skip=12pt,
  after skip=12pt
}

\usepackage[colorlinks,citecolor=blue,linkcolor=magenta,bookmarks=true]{hyperref}
\usepackage{doi}
\usepackage{cleveref}

\usepackage{physics2}
\usephysicsmodule{braket,ab}
\let\abs\vab
\let\norm\Vab

\hypersetup{pdftitle={Results on Fully Tolerant Product-State Testing}}
\usepackage{thm-restate}

\newtheorem{theorem}{Theorem}[section]
\newtheorem{claim}{Claim}[section]
\newtheorem{protocol}{Protocol}[section]
\newtheorem{fact}{Fact}[section]

\newtheorem{lemma}[theorem]{Lemma}
\newtheorem{corollary}[theorem]{Corollary}

\theoremstyle{remark}
\newtheorem{remark}[theorem]{Remark}

\newcommand{\C}{\mathbb C}
\newcommand{\Sym}{\operatorname{Sym}}
\newcommand{\Tr}{\operatorname{Tr}}
\newcommand{\Fprod}{F_{\mathrm{prod}}}

\newcommand{\1}{\mathbf 1}

\usepackage[toc,page]{appendix}

\usepackage{preamble}
\def\D{\mathsf{D}}
\def\wu{\widehat{u}}
\def\poly{\mathrm{poly}}
\def\Pihi{\Pi^{\mathrm{hi}}}
\def\ran{\operatorname{ran}}
\def\vepcolor{\vep_{\mathrm{color}}}
\def\E{\operatorname{\mathbb{E}}}
\def\xbar{\bar{x}}
\def\rhobar{\bar{\rho}}
\def\pproj{p_{\mathrm{proj}}}
\def\vepproj{\vep_{\mathrm{proj}}}
\def\pclone{p_{\mathrm{clone}}}
\def\vepclone{\vep_{\mathrm{clone}}}
\def\swap{\operatorname{SWAP}}
\newcommand{\symgroup}{S}

\title{Fully tolerant product state testing and closest product state
  learning}
\author{Zongbo Bao \thanks{Centrum Wiskunde \& Informatica (CWI) and QuSoft, Amsterdam, \texttt{zongbo.bao@cwi.nl}} \and Jonas Helsen \thanks{Centrum Wiskunde \& Informatica (CWI) and QuSoft, Amsterdam, \texttt{jonas@cwi.nl}} \and Tuyen Nguyen \thanks{University of Technology Sydney, Sydney, \texttt{Tuyen.Q.Nguyen@student.uts.edu.au }}  }
\date{October 1st 2026}
\begin{document}
\maketitle
\begin{abstract}
    We address the problem of testing whether an unknown $n$-qudit state $\rho$ is $a$-close to a product state or $b$-far away from any product state, as measured in terms of state overlap. We provide a time-efficient algorithm to solve this problem that requires an $n$-independent number of copies of the unknown state. Our random coloring argument shows that for any unknown state there always exists a partition of $[n]$ into $q$ parts such that the square of the overlap with the closest product state with respect to this partition is only an additive factor $O(1/q)$ larger. This reduces the problem to tolerant testing of $q$ parties, with potentially growing local dimensions. Combining this insight with blockwise spectral projection arguments we can show that the natural $k$-copy generalization of Harrow \& Montanaro's \cite{harrow2013testing} product state test provides an efficient tolerant tester. \\

    We use the same random coloring and blockwise spectral projection techniques to obtain a substantially improved algorithm for closest product state learning.
    In particular we give an algorithm that takes in $\widetilde{O}\big((nd)^2\big)\, 2^{\widetilde{O}(1/\varepsilon^8)}$ copies of the unknown state and produces an $\varepsilon$-approximately optimal product state. The key technical components of this learner are a qudit variant of the Bakshi et al. \cite{bakshi2024learningclosestproductstate}  high-fidelity product state learning algorithm and a sampling technique based on Werner's optimal cloning channel.
\end{abstract}

\section{Introduction}
A generic quantum state on $n$ qudits, $\rho \in \D\left((\mathbb{C}^d)^{\otimes n}\right)$, has a description whose size grows, in general, exponentially with $n$. An interesting problem is therefore to understand when additional structure permits one to infer useful information about an unknown state without per se reconstructing it completely. Perhaps the most basic such structure is \emph{product structure}. A pure product state has the form 
\begin{equation}
    \ket{\pi} = \ket{\pi_1}\otimes\cdots\otimes\ket{\pi_n},
\end{equation} 
for each $\ket{\pi_i} \in \C^d$. Product states are a fundamental class of states, which arise as the zero-entanglement limit of many structured state families, as mean-field approximations in many-body physics~\cite{brandao2013product}, and as unentangled witnesses in quantum complexity theory~\cite{harrow2013testing}. They therefore provide a natural setting in which to ask how global structure can be inferred from copies of an unknown quantum state.

There are two basic levels at which one may ask this question. At the decision level, one would like to \emph{test} whether an unknown state has the desired structure, and more generally to \emph{tolerantly test} whether its maximum fidelity with a product state is at least $a$ or at most $b$, for $a>b$. This is an instance of \emph{quantum property testing}~\cite{montanaro2013survey}, whose goal is to determine structural properties of an unknown quantum state using as few copies as possible. Harrow and Montanaro~\cite{harrow2013testing} gave an efficient algorithm for testing product structure in the non-tolerant setting, which was originally analyzed in the non-tolerant setting $a=1$. At the reconstruction level, one would instead like to \emph{learn} an explicit product state whose fidelity with the unknown state is close to the optimum. Bakshi et al.~\cite{bakshi2024learningclosestproductstate} gave an time-and-query efficient algorithm for this learning problem. Learning is therefore a strictly more informative task: a learner outputs a classical description of a product-state witness and can consequently be used to solve the corresponding testing problem, whereas a tester need only output a decision.
The distinction between testing and reconstruction is a recurring theme in quantum learning theory~\cite{montanaro2013survey, she2022unitarypropertytestinglower}. Analogous testing and learning questions arise for other structured families of quantum states, including stabilizer states~\cite{Gross_2021,grewal2024improvedstabilizerestimationbell,arunachalam2024polynomialtimetoleranttestingstabilizer, bao2025toleranttestingstabilizerstates, mehraban2024improvedboundstestinglow, arunachalam2026optimalstabilizertestinglearning} and matrix product states~\cite{soleimanifar2022testing, cramer2010efficient}, where testing and learning can exhibit substantially different complexities. Our goal is to understand the complexity of both tasks for product states, and in particular how it scales with the number of parties, the local dimension, and the desired accuracy.

We begin by formally introducing the testing and learning problems studied in this work. Throughout the paper, we use fidelity as our measure of closeness between quantum states. For a density operator $\rho\in\D\p{\p{\C^d}^{\otimes n}}$, define its maximum product fidelity as
\begin{equation}
    F_{\mathrm{prod}}(\rho)
    :=
    \max_{\substack{
        \ket{\pi_i}\in \mathbb{C}^{d}, \norm{\ket{\pi_i}}_2=1
     }}
   \bra{\pi_1}\otimes\cdots \otimes \bra{\pi_n}\rho\ket{\pi_1}\otimes\cdots\otimes\ket{\pi_n}.
    \label{eq:def-product-fidelity}
\end{equation}
Now, we are ready to state the testing and learning problems.
\begin{center}
\fbox{
\begin{minipage}{\linewidth}

Let $n\ge1$, $d\ge2$, and let the unknown state $\rho\in\D((\mathbb{C}^d)^{\otimes n})$.

\medskip

\noindent
\textbf{Problem 1.1 (Fully tolerant product-state testing).}
Fix thresholds $1\ge a>b\ge0$. Promised that $\rho$ satisfies either $ F_{\mathrm{prod}}(\rho)\ge a$ or $F_{\mathrm{prod}}(\rho)\le b$, decide which is the case using copies of $\rho$.
\medskip

\noindent
\textbf{Problem 1.2 (Agnostic product-state learning).}
Given an unknown state $\rho\in\D((\mathbb{C}^d)^{\otimes n})$ and an accuracy parameter $\varepsilon>0$, output a classical description of a pure product state $\ket{\widehat{\pi}}$ such that $\bra{\widehat{\pi}}\rho\ket{\widehat{\pi}}\ge F_{\mathrm{prod}}(\rho)-\varepsilon$,
using copies of $\rho$.
\end{minipage}}
\end{center}
For both problems, the algorithms are required to succeed with probability at least $1-\delta$, where $\delta\in(0,1/2)$ denotes the failure probability.

\subsection{Results}
We have two main results, which we state here. We give a technical overview of how these results are achieved in the next subsection.
\paragraph{Fully Tolerant Testing} The study of product-state testing was initiated by Harrow and Montanaro~\cite{harrow2013testing}, who introduced a remarkably simple two-copy test for the standard, non-tolerant setting. Their result established that productness can be tested with sample complexity independent of both the number of parties and their local dimensions using a permutation test\footnote{In the two-copy setting, it is equivalent to the standard SWAP test}. 
Subsequent works~\cite{soleimanifar2022testing, beckey2026optimalanalysisproducttest} progressively sharpened the analysis of this test, showing the optimal acceptance curve in every region of $\Fprod(\rho)$. 
However, the two-copy test yields an acceptance probability approaching $1/2$ as $\Fprod(\rho)\to 0$~\cite{beckey2026optimalanalysisproducttest} for some families of states. Thus, the two-copy test yields a nontrivial tolerant gap only for a restricted range of thresholds $a$ and $b$. This issue can be addressed by treating the tolerant tester as estimating $\Fprod(\rho)$ to additive accuracy of $O(a-b)$~\cite{bakshi2024learningclosestproductstate}, albeit with sample complexity that grows polynomially with the system size. It is therefore unclear whether the constant-sample complexity of product-state testing can be retained under tolerance. Such robustness is nontrivial: classically, there are properties that admit constant-query testers in the standard setting but require $n^{\Omega(1)}$ queries under tolerant testing~\cite{fischer2005tolerant}. A related phenomenon also appears for quantum states under structural relaxations: product states are precisely matrix product states of bond dimension one, whereas testing matrix product states of larger bond dimension requires a sample complexity of $\Omega(\sqrt{n})$~\cite{soleimanifar2022testing}. This raises the natural question of whether constant sample complexity nevertheless persists for product-state testing in the \emph{fully tolerant} setting, where even the YES instances may be entangled. Our first result answers this question affirmatively: for arbitrary thresholds $a>b$, we give a fully tolerant product-state tester whose sample complexity is independent of both the number of parties and the local dimension.
\begin{restatable}[Fully tolerant product-state tester]{theorem}{ProductTester}\label{thm:ProductTester}
    There is an algorithm that solves \textbf{Problem 1.1} with probability at least $1-\delta$ with sample complexity $2^{\widetilde{O}((a-b)^{-5})}
        \log\frac{2}{\delta}$ and computational complexity $2^{\widetilde{O}((a-b)^{-5})}\,n\,\poly\!\left(\log\frac{2nd}{\delta}\right).$
\end{restatable}
The key point is that full tolerance does not destroy the constant-sample phenomenon of product testing. A related constant-sample result can be obtained from the analysis of~\cite{she2022unitarypropertytestinglower}, which extends the inductive argument of~\cite{soleimanifar2022testing} to general $k$-copy permutation tests, but only for bipartite systems. In contrast, we propose a new analysis technique that applies to arbitrary $n$-partite systems and arbitrary local dimension, and yields a dimension-independent fully tolerant tester in this general setting.

\paragraph{Agnostic Learning} Our second result gives a direct fixed-parameter tractable algorithm for agnostic learning of product states.
\begin{restatable}[An FPT Agnostic Learner]{theorem}{ProductLearner}\label{thm:ProductLearner}
    There is a quantum algorithm that solves \textbf{Problem 1.2}, with probability at least $1-\delta$, with sample complexity of
    $2^{\widetilde O\p{\vep^{-8}}}(nd)^2
    \poly\log\f{nd}{\delta}$,
    and
    $2^{\widetilde O\p{\vep^{-8}}}n^3d^2
    \poly\log\f{nd}{\delta}$ computational complexity.
\end{restatable}
The same problem was studied by Bakshi et al.~\cite{bakshi2024learningclosestproductstate}, who achieve additive fidelity error $\vep$ with sample and computational complexity $n^{\poly(1/\vep)}$. Thus, for constant $\vep$, their algorithm runs in polynomial time, but the polynomial degree grows with $1/\vep$ and can become large even when $\vep$ is a small constant. In contrast, our learner achieves the same additive approximation guarantee with improved sample and computational complexity. Our new bound thereby moves the exponent dependence of $1/\vep$ into a prefactor while keeping the polynomial degree in the system size fixed. Consequently, for sufficiently large $n$ and small constant $\vep$, we obtain an asymptotic improvement in the dependence on the system size. Moreover, we also give the analysis for arbitrary local dimension $d\ge 2$ and it turns out that the sample and computational complexity have polynomial dependence on $d$. \footnote{The main result of Bakshi et al.~\cite{bakshi2024learningclosestproductstate} is stated for qubits. We do not view this as an inherent limitation of their approach; their techniques plausibly extend to higher local dimensions.}

Our computational complexity bound is also consistent with the hardness result of Bakshi et al.~\cite{bakshi2024learningclosestproductstate}, who show that estimating the optimal product-state fidelity to inverse-polynomial accuracy is $\mathsf{NP}$-hard. Consequently, unless $\mathsf{NP}\subseteq\mathsf{BQP}$, one cannot in general obtain an algorithm whose running time is polynomial in both the system size and $1/\vep$. Our result should therefore be viewed as a fixed-parameter tractable algorithm parameterized by $1/\vep$, with the dependence on the system size remaining polynomial.

\subsection{Technical Overview}
Although our tester and learner have different outputs, they rely on two complementary structural techniques that control the dependence on the number of subsystems $n$ and the local dimension $d$. The first technique, \textit{random coloring}, groups the subsystems into a bounded number of blocks while approximately preserving the optimal product fidelity. However, the blocks may still have large Hilbert-space dimensions. The second technique, \textit{blockwise spectral projection}, addresses this remaining difficulty by restricting each block to a high-eigenvalue subspace of its marginal, whose dimension depends only on the accuracy parameters. We first explain these shared techniques and then describe their respective roles in the testing and learning algorithms.

\paragraph{Random Coloring.} We independently assign each qudit a uniformly random color in $[q]$, obtaining a coloring $C:[n]\to[q]$ and the induced blocks $B_c=C^{-1}(c)$.
Let $F_C(\rho)$ denote the maximum fidelity of $\rho$ with a pure state that is product across these blocks:
\begin{equation}
    F_C(\rho) := \max_{\ket{\pi} = \ket{\pi_1} \otimes \dots \otimes \ket{\pi_q}} \bra{\pi}\rho\ket{\pi}.
\end{equation}
The question is whether we can choose the blocks so that this additional freedom barely increases the optimum. Our random coloring argument answers this affirmatively: for any $\rho$, there always exists a good partition $C$ such that
\begin{equation}
    \Fprod(\rho)^2 \leq F_C(\rho)^2 \leq \Fprod(\rho)^2 + \frac{3}{q}. \label{eqn: fc_vs_fprod}
\end{equation}
Indeed, the lower bound holds for any partition $C$. And the upper bound implies that the square of the overlap with the closest product state with respect to the good partition $C$ is only an additive factor $O(1/q)$ larger. Thus, for any additive accuracy $\eta>0$, a suitable partition into $O(\eta^{-2})$ blocks preserves the optimal product fidelity to within $\eta$, independently of $n$ and $d$.

\paragraph{Blockwise Spectral Projection.} The random coloring argument leaves a second difficulty: although there are only $q$ blocks, their Hilbert-space dimensions may still grow with both $n$ and $d$. To address this, let $\rho_c$ be the marginal on block $B_c$, and let $\Pi_c^{\mathrm{hi}} := \1_{[\theta,1]}(\rho_c)$ project onto its eigenvectors with eigenvalues at least a threshold $\theta>0$. 
This idea is also present in \cite{bakshi2024learningclosestproductstate}. Since $\Tr[\rho_c]=1$, the retained subspace $E_c := \ran\Pi_c^{\mathrm{hi}}$ satisfies
\begin{equation}
    \dim E_c \leq \frac{1}{\theta}.
    \label{eqn: bounded_dim}
\end{equation}
Moreover, the discarded part satisfies
\begin{equation}
    \norm{(I-\Pi_c^{\mathrm{hi}})\rho_c(I-\Pi_c^{\mathrm{hi}})}
    \leq \theta,
    \label{eqn: total_lows}
\end{equation}
The key intuition is that directions corresponding to very small eigenvalues of a block marginal cannot contribute substantially to the overlap with any global product state. Importantly,~\eqref{eqn: total_lows} controls the contribution along any individual direction in the low-eigenvalue subspace, even when that subspace carries substantial total weight. This motivates retaining the high-eigenvalue sectors, which capture the structure relevant to product fidelity up to a controlled error. By~\eqref{eqn: bounded_dim}, their dimensions are bounded solely in terms of $\theta$, independently of $n$ and $d$. We can therefore analyze the symmetric-subspace test on the resulting bounded-rank core while controlling the contribution of the discarded low-eigenvalue components separately.

Although sharing the same core components, these reductions are used differently in the two algorithms. For testing, the partition and spectral subspaces serve only as witnesses in the soundness analysis. The algorithm never constructs them. 
For learning, we sample a random coloring and implement blockwise spectral projection using the techniques in~\cite{WangZhang2025Samplizer, GSLW19}.
We explain these two applications next.

\subsubsection{Fully Tolerant Testing Algorithm}
Our tester follows a $k$-copy generalization of~\cite{harrow2013testing}, which is also studied in~\cite{she2022unitarypropertytestinglower} for bipartite states, and also implicitly considered in \cite{weinbrenner2026complete}. In particular, the $k$-copy generalization replaces each local SWAP test with a permutation test: for every qudit $i\in[n]$, it performs the projective measurement $\{Q_{k,i}, I-Q_{k,i}\}$ on the $k$ corresponding registers, where $Q_{k,i}$ projects onto $\operatorname{Sym}^k(\mathbb C^d)$. We denote its acceptance probability by
\begin{equation}
    Z_k(\rho)
    :=
    \Tr\!\left[
        \left(\prod_{i=1}^n Q_{k,i}\right)
        \rho^{\otimes k}
    \right].
    \label{eq:overview-product-test}
\end{equation}
The central analytical question is how tightly $Z_k(\rho)$ can be bounded in terms of the product fidelity $\Fprod(\rho)$. For $k=2$, the analysis from~\cite{beckey2026optimalanalysisproducttest, soleimanifar2022testing} indicates that the two-copy tester tends to accept with probability $1/2$ even as $F_{\mathrm{prod}}(\rho)\to0$ in the worst case. Consequently, these bounds do not yield a fully tolerant tester capable of resolving arbitrary fidelity gaps of $a$ and $b$. This motivates replacing the two-copy test with the general $k$-copy symmetric-subspace measurement and studying $Z_k(\rho)$. Specifically, we aim to prove a uniform sandwich bound:
\begin{equation}
F_{\mathrm{prod}}(\rho)
\le Z_k(\rho)^{1/k}
\le \alpha_k F_{\mathrm{prod}}(\rho)+\beta_k, \label{eqn: target_sandwich_bound}
\end{equation}
where $\alpha_k\to1$ and $\beta_k\to0$ are independent of the number of subsystems and their local dimensions. However, the analysis of $Z_k(\psi)$ was restricted to bipartite systems~\cite{she2022unitarypropertytestinglower}, and it is not obvious how to generalize it.

To overcome this difficulty, we analyze the fine-grained statistic $Z_k(\rho)$ through a good partition $C$. By~\Cref{eqn: fc_vs_fprod}, we can choose $C$ such that the corresponding block-product fidelity $F_C(\rho)$ gives an additive approximation to $\Fprod(\rho)$. It therefore remains to relate $Z_k(\rho)$ to $F_C(\rho)$. This is where the blockwise spectral projection becomes useful. By projecting each block onto its high-eigenvalue subspace and carefully controlling the contribution from the discarded low-eigenvalue sectors, we obtain
\begin{equation}
    Z_k(\rho)^{1/k}
    \leq
    F_C(\rho)+ 2q\theta,
\end{equation}
for sufficiently large $k$. Thus, after choosing $q$ so that the approximation guarantee in~\eqref{eqn: fc_vs_fprod} holds, we set the spectral threshold to $\theta=O(\eta/q)$ so that the error introduced by the projection is also at most $O(\eta)$. Combining these two bounds yields the desired upper bound on $Z_k(\rho)^{1/k}$ in terms of $\Fprod(\rho)$, establishing the sandwich bound in~\eqref{eqn: target_sandwich_bound}. Finally, following standard amplification arguments for tolerant testing~\cite{gs009, canonne2021pricetolerancedistributiontesting}, we obtain Theorem~\ref{thm:ProductTester}.

\subsubsection{Agnostic Product-State Learning Algorithm}
For learning, the same random coloring and blockwise spectral projection ideas must be used algorithmically rather than only as witnesses in the analysis. The main additional difficulty is that, given a coloring $C$, a state maximizing $F_C(\rho)$ is only guaranteed to be product across the color blocks, and may still be entangled within each block. It thus remains to recover a fully product state over the original $n$ qudits.
Note that for the analysis, at a desired level of accuracy $\vep$, it suffices to consider $\Fprod(\rho)\ge\vep$,
since otherwise every pure product state is already $\vep$-optimal.

In the nontrivial regime $\Fprod(\rho)\ge\vep$, we first choose a uniformly random coloring $C:[n]\to[q]$ and let $\ket{x_C}=\bigotimes_{c=1}^q\ket{x_c}$ achieve $F_C(\rho)$. For an appropriate choice of $q=\poly(1/\vep)$, our random coloring argument shows that, with constant probability, \begin{equation} 
    \sum_{c=1}^q\bigl(1-\Fprod(x_c)\bigr)=O(\vep), 
\end{equation} 
where $\Fprod(x_c)$ is defined with respect to the original qudits in block $B_c$. Thus the block optimizer is nearly fully product, with total infidelity of order $\vep$ across all blocks. This quantitative guarantee will control the final approximation error.

We next apply the blockwise spectral projection introduced above. Unlike in the tester, these spectral projectors must actually be implemented. We do this from copies of the marginals $\rho_c$ using spectral filtering and the sample-based simulation technique of~\cite{GSLW19,WangZhang2025Samplizer}. For each block marginal $\rho_c$, we implement the projector $\Pi_c^{\mathrm{hi}}$ onto its high-eigenvalue subspace $E_c$. By choosing the threshold appropriately, $\dim E_c=\poly(1/\vep)$ is independent of $n$ and $d$, while the projection approximately preserves the corresponding optimizer $\ket{x_c}$. 

However, although the dimension of $E_c$ is bounded, learning a candidate state within this subspace remains challenging because we do not have access to an explicit basis for $E_c$. To overcome this, we use Werner’s optimal cloning channel~\cite{Werner1998}. For each color $c$, we apply the cloning channel on the bounded-dimensional subspace $E_c$. For a sufficiently large number of output copies, the joint output of these cloning channels is close to a measure-and-prepare channel of the form
\begin{equation}
    \mathbb{E}_{(v_1,\ldots,v_q)\sim\nu_C}
    \left[
        \bigotimes_{c=1}^q
        \bigl(\ket{v_c}\!\bra{v_c}\bigr)^{\otimes L}
    \right],
\end{equation}
for some joint distribution $\nu_C$ over unit vectors $\ket{v_c}\in E_c$. This description is used only in the analysis. The algorithm feeds the physical clone registers directly to the next step. 
The important property is that, with probability at least $\exp(-\poly(1/\vep))$, the sampled vectors satisfy
\begin{equation}
  \sum_{c=1}^q
  \bigl(1-|\langle\bar{x}_c|v_c\rangle|^2\bigr)=O(\vep),
\end{equation}
where $\ket{\bar{x}_c}$ is the normalized projection of $\ket{x_c}$ onto $E_c$. Together with the coloring and projection guarantees, this gives $\sum_c(1-\Fprod(v_c))=O(\vep)$ and keeps the sampled vectors close to the particular block optimizer $\ket{x_C}$.
To our knowledge this is the first time approximate cloning was used in an algorithmic fashion. Because the argument is quite generic we expect it to be a useful subroutine for agnostic tomography for other classes of states. 

These guarantees play two distinct roles. To apply the qudit high-fidelity product-state learner from~\cite{bakshi2024learningclosestproductstate} to $\ket{v_c}$ with additive accuracy $\vep_{\mathrm{hfl}}$, it suffices that $\Fprod(v_c)\ge5/6+\vep_{\mathrm{hfl}}$. Our parameter choices ensure this promise for every block on the good sampling event. We also use the stronger aggregate infidelity bounds to control the accuracy of the full learner. Choosing $q\vep_{\mathrm{hfl}}=O(\vep)$, the resulting product states $\ket{z_c}$ satisfy
\begin{equation}
  e_{\mathrm{out}}:=\sum_c\bigl(1-|\langle z_c|x_c\rangle|^2\bigr) =O(\vep).
\end{equation}
The optimality of $\ket{x_C}$ then implies $\bra z\rho\ket z\ge\Fprod(\rho)-O(\vep)$ for $\ket z=\bigotimes_c\ket{z_c}$, as shown in Lemma~\ref{lem:learner-output-quality}.
We therefore repeat the candidate-generation procedure sufficiently many times, estimate the fidelity of every candidate using fresh copies of $\rho$, and return the candidate with the largest estimated fidelity. With the appropriate parameter choices, this yields an $\vep$-optimal product state with probability at least $1-\delta$, giving the fixed-parameter tractable learner of Theorem~\ref{thm:ProductLearner}.

\subsection{Use of AI Tools}

This work was developed with assistance from ChatGPT, using the
5.6 Sol, 6.0 Astra, and 6.0 Sol models. We (the authors) formulated the
research questions and directed the development of the testing and
learning approaches. We used these tools to explore examples,
generate and revise proof drafts, develop implementation and
complexity analyses, run numerical checks, and assisted in drafting the manuscript.

We started thinking about this problem over a year ago, with the idea that the $k$-copy generalization of Harrow \& Montanaro's product test should be a tolerant tester at constant $k$. We didn't get anywhere, and for a while we thought we should be able to construct a lower bound by adapting the hardness argument in \cite{bakshi2024learningclosestproductstate} (which essentially reduces a four-party product fidelity estimation problem (with growing local dimensions) to a computationally hard problem). With TN visiting CWI over the summer we restarted seriously exploring this problem. Some back and forth about a query-controlled version of this argument led 5.6 Sol to propose a fully tolerant product-state
tester for tripartite states, which explicitly implemented a
blockwise spectral projection (thus refuting our lower bound route). Related ideas appear in the bipartite
analysis of~\cite{she2022unitarypropertytestinglower}. The authors
observed that the projection could be used solely in the analysis,
yielding a uniform bound for the registerwise permutation test.
This observation underlies \Cref{thm:explicit-spectral-core}.
We also explored direct estimates from the Haar integral
representation of the permutation-test acceptance probability.
We subsequently focused on approximating the fully product fidelity
by the product fidelity associated with a suitable partition into
a bounded number of blocks. This was very successful in the case of permutation invariant states but a general argument initially eluded us (even with AI assistance).

We initially proposed random coloring as a good method for choosing a partition in the general case, which the initial
AI responses assessed as unpromising. Subsequent exploration
included a proposed nonconstructive argument using ultrafilters
and ultraproducts to establish a copy bound independent of the
system size for fixed testing thresholds. This argument induced us to consider a regularity-lemma approach, producing
an intermediate argument with sample complexity bounded by a tower
of twos of height $\poly(1/\Delta)$, where $\Delta=a-b$ is the testing
gap. As an intermediate result, we attempted to a doubly exponential bound of the form
$2^{2^{\poly(1/\Delta)}}$, using a type of dimension-independent Helly theorem from \cite{adiprasito2020theorems} (which governs intersection sizes for sets of convex sets). Eventually, AI developed
our initial random-coloring idea into the general partition
argument used in the present proof, leading to the sample
complexity stated in \Cref{thm:ProductTester}.

For learning, we proposed combining random coloring,
an implementation of blockwise spectral projection, and the high-fidelity
product-state learner of~\cite{bakshi2024learningclosestproductstate}.
Starting from this outline, AI generated the initial versions of
the remaining construction and proof details. AI also developed
the proposed extension of the high-fidelity learner to qudits. This extension is not particularly surprising, largely consisting of a careful retracing of the argument in ~\cite{bakshi2024learningclosestproductstate}. More surprising is the use of the Werner approximate cloning channel in an algorithmic fashion, which was entirely developed by AI and for which we found no precedent for in the literature. 

We subsequently reviewed and substantially revised the
proofs and exposition for mathematical clarity and readability.
This included reorganizing the lemmas, standardizing notation,
and identifying and adding appropriate references. We (obviously)
take responsibility for the final statements and proofs.

Finally we want to note that neither 5.6 Sol  nor 6.0 Astra  was able (at the time of testing, early September 2026) to re-derive either the tester or the learner from a generic prompt, even with substantial token use. We thus consider this a successful combination of human ideas and artificially generated argumentation. 

\subsection*{Acknowledgements}
The authors want to thank Amira Abbas for introducing TN and ZB and kickstarting this project. The authors also thank Cl\'ement Canonne and Aditya Vikram Singh for helpful discussions on tolerant testing.
JH acknowledges funding from the Dutch Research Council
(NWO) through a Veni grant (grant No.VI.Veni.222.331) and the Quantum Software Consortium
(NWO Zwaartekracht Grant No.024.003.037).
ZB acknowledges support from QDNL (EZK Groeifonds) through the Cat-2 program. Part of this project was undertaken while TN was a PhD intern at Centrum Wiskunde \& Informatica (CWI) in Amsterdam (June-September).
\section{Preliminaries}
\subsection{Basic Notation}
\label{sec:basic-notation}
For a positive integer $m$, we write $[m]:=\{1,\ldots,m\}$.
We write $\symgroup_k$ for the symmetric group over $k$ elements.
All Hilbert spaces are finite-dimensional and complex, and
$\mathcal H_i=\C^d$ denotes the space of the $i$-th qudit.
We write $\D(\mathcal H)$ for the density operators on
$\mathcal H$, namely positive semidefinite operators with trace one.
A pure state is represented by a unit vector $\ket{v}$.
When a pure state appears as the argument of a function on density
operators, as in $\Fprod(v)$, we identify $v$ with
$\ket{v}\!\bra{v}$.

For Hermitian operators, $A\ge B$ means that $A-B$ is positive
semidefinite. We write $I$ for the identity operator and omit
identity tensor factors when the registers are clear from context.
For a block $B_c$, we write $\rho_c=\Tr_{-c}\rho$ for its marginal,
where $\Tr_{-c}$ traces out all blocks other than $c$.
We also use $\Tr_{-B_c}$ for this partial trace.

For vectors, $\|\cdot\|_2$ denotes the Euclidean norm.
For operators, the trace norm and operator norm are
\begin{equation}
  \|A\|_1:=\Tr\sqrt{A^\dagger A},
  \qquad
  \|A\|_\infty:=\sup_{\|v\|_2=1}\|Av\|_2.
\end{equation}
An operator norm without a subscript means $\|\cdot\|_\infty$.
The trace distance between $\rho$ and $\sigma$ is
$\frac12\|\rho-\sigma\|_1$.
For linear maps, $\|\cdot\|_\diamond$ denotes the diamond norm.

We use $\ran P$ to denote the range
of a projector $P$. For a Hermitian operator
$A=\sum_j\lambda_j\ket{v_j}\!\bra{v_j}$ and an interval
$J\subseteq\mathbb R$, its spectral projector onto $J$ is
\begin{equation}
  \1_J(A):=\sum_{j:\lambda_j\in J}\ket{v_j}\!\bra{v_j}.
\end{equation}
For a Hilbert space $E$, $\Sym^k(E)$ denotes the subspace of
$E^{\otimes k}$ invariant under all permutations of the tensor
factors.

Unless otherwise specified, $\log$ denotes the natural logarithm;
binary logarithms are written as $\log_2$.
We use $\poly$ to denote polynomial dependence and
$\widetilde O$ to suppress polylogarithmic factors in the
relevant parameters.

\subsection{$k$-copy Product State Tester}
For an integer $k\ge1$ and each $i\in[n]$, let $Q_{k,i}$ denote
the orthogonal projector onto
$\Sym^k(\mathcal H_i)\subseteq\mathcal H_i^{\otimes k}$.
It acts on the $i$-th register of each of the $k$ copies and
as the identity on all remaining registers. Define
\begin{equation}
  Q_k:=\prod_{i=1}^n Q_{k,i}.
\end{equation}
These local projectors commute because they act on disjoint
registers. Thus, $Q_k$ is the accepting projector of the test
that checks whether every register lies in its local symmetric
subspace.

We begin with the natural $k$-copy generalization of the product test from Harrow and Montanaro \cite{harrow2013testing}.
\begin{protocol}[$k$-copy product test]\label{prot:k-copy-product-test}
Given independent copies of a density operator
$\rho\in\D\p{\p{\C^d}^{\otimes n}}$,
the $k$-copy product test proceeds as follows.

\begin{enumerate}
    \item Prepare the joint input state $\rho^{\otimes k}$.

    \item For each register $i\in[n]$, perform the projective measurement
    $\left\{Q_{k,i}, I-Q_{k,i}\right\}$ on its $k$ registers.

    \item Accept if every measurement returns the symmetric outcome
    $Q_{k,i}$. Otherwise, reject.
\end{enumerate}
\end{protocol}

\noindent The acceptance probability of Protocol~\ref{prot:k-copy-product-test} is
\begin{equation}
  Z_k(\rho):=\Tr\!\left[
    Q_k\rho^{\otimes k}
  \right].
  \label{eq:k-copy-acceptance}
\end{equation}
Every pure product vector $\ket{\pi}$ has $\ket{\pi}^{\otimes k}$ in the accepting subspace. The accepting projector therefore
dominates $(\ket{\pi}\!\bra{\pi})^{\otimes k}$, and maximizing over product vectors yields the completeness bound. Thus, we have the following universal lower bound:
\begin{claim}
For any $k\geq 2$, $Z_k(\rho)\ge \Fprod(\rho)^k$.
\end{claim}
\section{A Fully Tolerant Tester}\label{sec: tester}

\subsection{The $k$-copy Product Test}

Our approach introduces an intermediate optimization over partitions
with a bounded number of blocks. For a coloring $C:[n]\to[q]$,
let $B_c:=C^{-1}(c)$ and
$\mathcal H_{B_c}:=\bigotimes_{i\in B_c}\mathcal H_i$ for each
$c\in[q]$.
When it is clear from the context we also use $\mathcal H_c$ for $\mathcal H_{B_c}$.
If $B_c=\varnothing$, we use the empty tensor product
convention $\mathcal H_{c}=\C$. Define
\begin{equation}
  F_C(\rho):=
  \max_{\substack{
    \ket{v_c}\in\mathcal H_{c},\ \norm{\ket{v_c}}_2=1\\
    c\in[q]
  }}
  \bra{v_1\otimes\cdots\otimes v_q}
  \rho
  \ket{v_1\otimes\cdots\otimes v_q}.
  \label{eq:def-FC}
\end{equation}
We then minimize over all such colorings, with at most $q$ colors:
\begin{equation}
  F_q(\rho):=
  \min_{C:[n] \to [q]}
  F_C(\rho).
  \label{eq:def-Fq}
\end{equation}
The optimization defining $F_C$ allows entanglement within each block
while requiring product structure across blocks. Every fully product
vector is admissible, so we have the following fact.
\begin{fact}\label{fact1}
For every coloring $C$ with at most $q$ nonempty blocks,
$\Fprod(\rho)\le F_q(\rho)\le F_C(\rho)$.
\end{fact}

Our analysis proceeds in two steps. In
Subsection~\ref{sec:uniform-bounded-partition}, we use a random
$q$-coloring to show that
$F_q(\rho)^2-\Fprod(\rho)^2\le3/q$.
Thus a bounded number of color blocks suffices to approximate the
fully product fidelity, independently of $n$ and $d$.
In Subsection~\ref{sec:chosen-partition-fidelity}, we choose a coloring
attaining $F_q(\rho)$ and apply blockwise spectral projection onto the
high-eigenvalue subspaces of the block marginals. The retained operator has bounded local
support dimensions, so its contribution to the test is controlled by
$F_q(\rho)$ and the dimensions of the corresponding symmetric
subspaces. The discarded components have small marginal eigenvalues,
which bounds their contribution. Choosing $q$, the spectral cutoff,
and $k$ appropriately therefore makes $Z_k(\rho)^{1/k}$ arbitrarily close to
$\Fprod(\rho)$, uniformly in the system size and local dimension.

Finally, in Subsection~\ref{subsec:tester-complete}, we combine these two results to separate the
one-trial acceptance probabilities in the two promised cases.
Repeating the test and thresholding the empirical acceptance
frequency gives the fully tolerant tester and its resource bounds.
Random coloring and blockwise spectral projection are used only in the analysis;
the tester performs Protocol~\ref{prot:k-copy-product-test}.

\subsection{Random Coloring}\label{sec:uniform-bounded-partition}
Our first step is to show that coarse-graining the $n$ registers into a bounded number of blocks increases the optimal product fidelity by only a controlled amount, which is independent of both the number of subsystems $n$ and the local dimension $d$. The following theorem quantifies
this approximation.

\begin{theorem}\label{thm: random_color tester}
    Given $n\geq 1$, $d\geq 2$, for every integer $q \geq 1$ and any mixed state $\rho\in\D\p{\p{\C^d}^{\otimes n}}$, the following inequality holds
    \begin{equation}
        \Fprod(\rho)^2 \leq F_q(\rho)^2 \leq \Fprod(\rho)^2 + \frac{3}{q} 
    \end{equation}
\end{theorem}
Theorem~\ref{thm: random_color tester} shows that $F_q$ provides a uniform additive approximation to $\Fprod(\rho)$, with an error that vanishes
as the number of blocks $q$ increases.  
We have the following corollary.
\begin{corollary}
\label{cor:fq_vs_ffull}
For any $0<\vepcolor\le1$, choose $q=\left\lceil 3/\vepcolor\right\rceil$.
Then, for any mixed state $\rho\in\D\p{\p{\C^d}^{\otimes n}}$,
\begin{equation}
0\le F_q(\rho)^2-\Fprod(\rho)^2 \le\vepcolor.
\end{equation}
\end{corollary}

To compare $F_C(\rho)$ with $\Fprod(\rho)$, we use the two-copy
acceptance probability as an intermediate quantity. Let $Q_{2,c}$ be
the symmetric-subspace projector for two copies of color block $B_c$, and set
$Q_{2,C}:=\prod_{c\in [q]}Q_{2,c}$ and
$Z_{2,C}(\rho):=\Tr(Q_{2,C}\rho^{\otimes2})$.
The ordinary two-copy test introduced by Harrow and Montanaro \cite{harrow2013testing} corresponds to the case where each register has a unique color.
Its accepting projector is contained in $Q_{2,C}$, so
$Z_{2,C}(\rho)-Z_2(\rho)\ge0$.

The difference $Z_{2,C}(\rho)-Z_2(\rho)$ measures the acceptance lost
when the color blocks are refined into individual registers. For a random
coloring this loss is small on average. The proof decomposes the
two-copy space according to the swaps of corresponding registers. A
sector rejected by the registerwise test can pass the blockwise test only
when the negative swap signs pair within color blocks.

For a fixed coloring, the same loss controls the increase in product
overlap allowed by grouping registers into blocks. We refine the blocks
one at a time and show that the decrease in squared block-product
overlap is at most three times the loss in acceptance. Adding these
bounds controls $F_C(\rho)^2-\Fprod(\rho)^2$ by the total test gap.
The next lemma states both estimates, which naturally implies \Cref{thm: random_color tester}. 
We note that the part $(i)$ of the lemma can also be obtained by \Cref{cor:direct-good-coloring}, where in the learning setting we need a substantially stronger result. In the interest of self-containedness, we give a direct proof here.

\begin{lemma}
\label{lem:one-block-leakage}
For any mixed state $\rho\in\D\p{\p{\C^d}^{\otimes n}}$, the following hold.

\begin{enumerate}[label=(\roman*)]
\item If $C:[n]\to[q]$ assigns each register an independent
uniformly random color, then
\begin{equation}
\mathbb E_C[Z_{2,C}(\rho)-Z_2(\rho)]
\le \frac{1}{q}.
\end{equation}

\item For any coloring $C$,
\begin{equation}
Z_{2,C}(\rho)-Z_2(\rho)
\ge\frac{F_C(\rho)^2-\Fprod(\rho)^2}{3}.
\end{equation}
\end{enumerate}
\end{lemma}

\begin{proof}
\noindent\textit{Part (i).}
For each register $i$, let $\swap_i$ be the operator which swaps the $i$-th register between the two copies.
Then $Q_2=\prod_{i\in[n]}(I+\swap_i)/2$ is the registerwise accepting
projector, while $Q_{2,c}=(I+\prod_{i\in B_c}\swap_i)/2$ for each $c\in[q]$.
The $\swap_i$ commute, so their joint eigenspaces are indexed by sets
$A\subseteq[n]$: on the eigenspace indexed by $A$, the swap $\swap_i$ has
eigenvalue $-1$ exactly when $i\in A$. 
One important observation is that the projector $Q_2$ accepts only the eigenspace indexed by 
$A=\varnothing$, whereas the projector $Q_{2,C}$ accepts the eigenspace indexed
by $A$ exactly when $|A\cap B_c|$ is even for every $c\in[q]$.

Fix a nonempty $A$ and one register $i\in A$, and condition on
the colors of the other registers in $A$. At most one color
for $i$ makes all the counts even. We conclude that
\begin{equation}
\Pr_C\!\left[|A\cap B_c|\text{ is even for every }c\in[q]\right]
\le\frac1q.
\end{equation}
Thus $\mathbb E_C[Q_{2,C}-Q_2]$ acts as zero on the eigenspace indexed
by $A=\varnothing$ and as a scalar at most $1/q$ on every other joint
eigenspace. Hence
\begin{equation}
\mathbb E_C[Q_{2,C}-Q_2] \le \frac{I-Q_2}{q}.
\label{eq:lemma33-random-operator-bound}
\end{equation}
Taking the trace against $\rho^{\otimes2}$ gives the desired bound
\begin{equation}
\mathbb E_C[Z_{2,C}(\rho)-Z_2(\rho)]
\le\frac{1-Z_2(\rho)}q\le\frac1q.
\end{equation}

\noindent\textit{Part (ii).}
We use the pure-state two-copy product-test bound
\begin{equation}
1-Z_2(v)\ge\frac{1-\Fprod(v)^2}{3}
\label{eq:lemma33-pure-bound}
\end{equation}
for every pure multipartite state $\ket v$; it follows directly from Theorem~1.1 of~\cite{beckey2026optimalanalysisproducttest}.
Fix a coloring $C:[n]\to[q]$ and a color $c\in[q]$.
Let $C^{(c)}$ be obtained by giving the registers in $B_c$ distinct
new colors, leaving the other color classes unchanged.
Choose a color-product vector
$\ket{x_c}\otimes\ket{x_{-c}}$ achieving $F_C(\rho)$, where
$\ket{x_c}$ is on $B_c$ and $\ket{x_{-c}}$ is product across the other
color classes. With $\ket{x_{-c}}$ fixed, optimality makes $\ket{x_c}$
a top eigenvector of the positive operator
$\bra{x_{-c}}\rho\ket{x_{-c}}$ on $B_c$, with eigenvalue $F_C(\rho)$.
Therefore
\begin{equation}
\bra{x_{-c}}\rho\ket{x_{-c}}
\ge F_C(\rho)\ket{x_c}\bra{x_c}.
\end{equation}

Let $Q^{\mathrm{fine}}_{2,c}:=\prod_{i\in B_c}(I+\swap_i)/2$ be the registerwise two-copy test in the color block $B_c$,
recall that $Q_{2,c} = \p{I+\prod_{i\in B_c} \swap_i}/2$ is the two-copy swap test for the whole color block $B_c$. We have $Q_{2,c}\ge Q^{\mathrm{fine}}_{2,c}$.
We also note that 
\begin{equation}
    Q_{2,C}-Q_{2,C^{(c)}} = \left(Q_{2,c}-Q^{\mathrm{fine}}_{2,c}\right)
\left(\prod_{c'\in[q]\setminus\{c\}}Q_{2,c'}\right)
    \ge \left(Q_{2,c}-Q^{\mathrm{fine}}_{2,c}\right)\otimes \ketbra{x_{-c}}{x_{-c}}^{\otimes 2}
\end{equation}
Taking the trace against $\rho^{\otimes2}$ and using the conditional operator bound on each copy, we obtain

\begin{align}
Z_{2,C}(\rho)-Z_{2,C^{(c)}}(\rho)
&=\Tr\!\left[\p{Q_{2,C}-Q_{2,C^{(c)}}}\rho^{\otimes2}\right]\\
&\ge\Tr\!\left[\left(Q_{2,c}-Q^{\mathrm{fine}}_{2,c}\right)
\left(\bra{x_{-c}}\rho\ket{x_{-c}}\right)^{\otimes2}\right]\\
&\ge F_C(\rho)^2\bra{x_c}^{\otimes2}
\left(Q_{2,c}-Q^{\mathrm{fine}}_{2,c}\right)
\ket{x_c}^{\otimes2}\\
&=F_C(\rho)^2\bigl(1-Z_2(x_c)\bigr)\\
&\ge\frac{F_C(\rho)^2}{3}(1-\Fprod(x_c)^2)\label{eq:one-color-refinement-leakage},
\end{align}
where the last step uses
\eqref{eq:lemma33-pure-bound}. If $\ket{u_c}$ is a product vector on
$B_c$ with $|\langle u_c|x_c\rangle|^2=\Fprod(x_c)$, then
\begin{equation}
F_{C^{(c)}}(\rho)
\ge\bra{u_c\otimes x_{-c}}\rho\ket{u_c\otimes x_{-c}}
\ge F_C(\rho)\Fprod(x_c).
\end{equation}
Thus this one-color refinement satisfies
\begin{equation}
Z_{2,C}(\rho)-Z_{2,C^{(c)}}(\rho)
\ge \frac{F_C(\rho)^2}{3}(1-\Fprod(x_c)^2)
\ge\frac{F_C(\rho)^2-F_{C^{(c)}}(\rho)^2}{3}.
\end{equation}

Set $C^{[0]}=C$. For each $c\in[q]$, let $C^{[c]}$ be obtained from
$C^{[c-1]}$ by refining color $c$ into singleton classes; if that
class is already a singleton or empty, leave the coloring unchanged.
Then $C^{[q]}$ is the singleton coloring, and the one-color bound gives
\begin{align}
Z_{2,C}(\rho)-Z_2(\rho)
&=\sum_{c=1}^{q}\bigl(Z_{2,C^{[c-1]}}(\rho)-Z_{2,C^{[c]}}(\rho)\bigr)\notag\\
&\ge\frac13\sum_{c=1}^{q}
\bigl(F_{C^{[c-1]}}(\rho)^2-F_{C^{[c]}}(\rho)^2\bigr)\notag\\
&=\frac{F_C(\rho)^2-\Fprod(\rho)^2}{3}.
\end{align}
\end{proof}

\subsection{Blockwise Spectral Projection}\label{sec:chosen-partition-fidelity}
The preceding subsection guarantees the existence of a partition with controlled product fidelity. We now relate the test statistic $Z_k(\rho)$ to the partition-product fidelity $F_q(\rho)$. The main idea is to project each block onto its high-eigenvalue subspace,
to make the state have bounded local ranks while bounding the error introduced by this step.
The intuition is that directions corresponding to very small eigenvalues of a block marginal cannot contribute substantially to the overlap with any global product state. Thus, product fidelity is governed primarily by the high-eigenvalue sectors of the local marginals rather than by the total weight distributed across their low-eigenvalue subspaces. Since a density operator can have only a bounded number of eigenvalues above any fixed threshold, restricting to these sectors yields an effective state with bounded local dimensions. This reduction allows us to analyze the symmetric-subspace test on the bounded-rank core, while treating the discarded low-eigenvalue components as a controlled error.

In particular, we establish the following theorem.

\begin{theorem}\label{thm:explicit-spectral-core}
Let $n\ge1$, $d\ge2$, and
$\rho\in\D\!\left((\C^d)^{\otimes n}\right)$. For integers
$q,k\ge1$ and $R\ge2$,
\begin{equation}
Z_k(\rho)^{1/k}
\le
\dim\Sym^k(\C^R)^{q/k}F_q(\rho)
+\frac{q}{R}\dim\Sym^k(\C^R)^{1/k},
\label{eq:master-partition}
\end{equation}
where $\dim\Sym^k(\C^R) = \binom{k+R-1}{k}$.
\end{theorem}

The first term comes from the part supported on at most $R$ dimensions
per color. Its bound uses the same Haar-integral representation of
symmetric projectors as in~\cite[Appendix~A]{weinbrenner2026complete}. The second term bounds what is discarded by the blockwise spectral projection. We make the
second term small by choosing $R$ of order $q/\vep_{\mathrm{sp}}$ and
then choose $k$ so that the finite-rank factor is close to one, which gives us the following corollary.

\begin{corollary}[Parameters for blockwise spectral projection]
\label{cor:fidelity-sensitive-core}
Let $n,d,q$ be integers with $n\ge1$, $d\ge2$, $q\ge1$, and
let $\rho\in\D\!\left((\C^d)^{\otimes n}\right)$. For
$0<\vep_{\mathrm{sp}}\le1$, set
\begin{equation}
R:=\left\lceil\frac{4q}{\vep_{\mathrm{sp}}}\right\rceil,
\qquad
k_{\mathrm{sp}}
:=
\left\lceil
\frac{8qR\max\{F_q(\rho),\vep_{\mathrm{sp}}\}}
{\vep_{\mathrm{sp}}}
\log\!\left(
\frac{8eq\max\{F_q(\rho),\vep_{\mathrm{sp}}\}}
{\vep_{\mathrm{sp}}}\right)
\right\rceil.
\end{equation}
For every integer $k\ge k_{\mathrm{sp}}$,
\begin{equation}
Z_k(\rho)^{1/k}\le F_q(\rho)+\vep_{\mathrm{sp}}.
\label{eq:fidelity-sensitive-additive}
\end{equation}
\end{corollary}

\begin{proof}
Set $s:=8q\max\{F_q(\rho),\vep_{\mathrm{sp}}\}/\vep_{\mathrm{sp}}\ge8$.
The choice of $k$ gives $k/R\ge s\log(es)\ge1$. The binomial bound
$\dim\Sym^k(\C^R)\le(2ek/R)^R$ and the fact that
$x\mapsto\log(2ex)/x$ decreases for $x\ge1$ give
\begin{equation}
\begin{aligned}
\frac{q}{k}\log\dim\Sym^k(\C^R)
&\le\frac{qR}{k}\log(2ek/R)\\
&\le\frac{q\log(2es\log(es))}{s\log(es)}\\
&\le\frac{2q}{s}
=\frac{\vep_{\mathrm{sp}}}{4\max\{F_q(\rho),\vep_{\mathrm{sp}}\}}
\le\frac14.
\end{aligned}
\end{equation}
The third inequality uses $2\log(es)\le es$.
Using $e^x\le1+2x$ for $0\le x\le1/4$, we obtain
\begin{equation}
\bigl(\dim\Sym^k(\C^R)\bigr)^{q/k}
\le1+\frac{\vep_{\mathrm{sp}}}{2\max\{F_q(\rho),\vep_{\mathrm{sp}}\}}
\le2.
\end{equation}
Since $q\ge1$, the factor
$\bigl(\dim\Sym^k(\C^R)\bigr)^{1/k}$ is also at most $2$.
Thus \Cref{thm:explicit-spectral-core} and
$q/R\le\vep_{\mathrm{sp}}/4$ give
\begin{equation}
\begin{aligned}
Z_k(\rho)^{1/k}
&\le F_q(\rho)
+\frac{\vep_{\mathrm{sp}}F_q(\rho)}{2\max\{F_q(\rho),\vep_{\mathrm{sp}}\}}
+\frac{2q}{R}\\
&\le F_q(\rho)+\vep_{\mathrm{sp}}.
\end{aligned}
\end{equation}
\end{proof}

For the proof of \Cref{thm:explicit-spectral-core}, use the marginal of $\rho$ on each color class to select
its eigenspaces above the threshold $1/R$. Each selected eigenspace
has dimension less than $R$. After all colors have been compressed,
the test on the remaining operator is controlled by $F_C(\rho)$.
Each discarded operator has a marginal bounded by $I/R$ on the color
where it was discarded.

For a fixed coloring $C$, let $Q_{k,c}$ denote the block projector onto
$\Sym^k(\mathcal H_c)$ on the $k$ copies of block $B_c$,
acting as the identity on all other blocks. Set
$Q_{k,C}:=\prod_{c\in[q]}Q_{k,c}$ and
$Z_{k,C}(\rho):=\Tr[Q_{k,C}\rho^{\otimes k}]$.
This is the blockwise test's acceptance probability, whereas
$Z_k(\rho)=\Tr[Q_k\rho^{\otimes k}]$ is its registerwise counterpart.

The following lemma makes these estimates
precise for mixed inputs,
which naturally implies \Cref{thm:explicit-spectral-core}.

\begin{lemma}\label{lem:mixed-spectral-core}
Fix $C:[n]\to[q]$, $R\ge2$, $k\ge1$, and
$\rho\in\D\!\left((\C^d)^{\otimes n}\right)$.
For each $c\in[q]$, let
$\rho_c:=\Tr_{-c}\rho$,
$\Pihi_c:=\mathbf1_{(1/R,1]}(\rho_c)$ be the projector to the eigenspaces of $\rho_c$ with eigenvalues larger than $1/R$,
$\Pi^{\mathrm{lo}}_c:=I-\Pihi_c$, and
$\Pihi_{<c}:=\prod_{j<c}\Pihi_j$, with $\Pihi_{<1}=I$.
Set $\omega_0^{\mathrm{hi}}:=\rho$ and
\begin{equation}
\omega_c^{\mathrm{hi}}
:=\Pihi_c\Pihi_{<c}\rho\Pihi_{<c}\Pihi_c,
\qquad
\omega_c^{\mathrm{lo}}
:=\Pi^{\mathrm{lo}}_c\Pihi_{<c}\rho\Pihi_{<c}\Pi^{\mathrm{lo}}_c.
\end{equation}
Note that $\omega_q^{\mathrm{hi}} = \p{\bigotimes_c \Pihi_c} \rho \p{\bigotimes_c \Pihi_c}$.
Then:
\begin{enumerate}[label=(\roman*)]
\item The successive compressions satisfy
\begin{equation}
Z_k(\rho)^{1/k}\le 
Z_{k,C}(\rho)^{1/k}
\le
Z_{k,C}(\omega_q^{\mathrm{hi}})^{1/k}
+\sum_{c=1}^{q}Z_{k,C}(\omega_c^{\mathrm{lo}})^{1/k}.
\end{equation}
\item The high- and low-eigenspace parts obey
\begin{equation}
\begin{aligned}
Z_{k,C}(\omega_q^{\mathrm{hi}})^{1/k}
&\le \dim\Sym^k(\C^R)^{q/k}F_C(\rho),\\
Z_{k,C}(\omega_c^{\mathrm{lo}})^{1/k}
&\le \frac{\dim\Sym^k(\C^R)^{1/k}}{R}
\qquad \forall c\in[q].
\end{aligned}
\end{equation}
\end{enumerate}
\end{lemma}

\begin{proof}
\noindent\textit{Part (i).}

Let $\tau$ be any positive operator on all the registers and fix $c\in[q]$.
For any positive operator $A$ on $\mathcal H_c$, define
\begin{equation}
h_k(A):=\Tr(Q_{k,c}A^{\otimes k})
=\dim\Sym^k(\mathcal H_c)
\int\langle u_c|A|u_c\rangle^k\,d\mu_c(u_c),
\end{equation}
where $\mu_c$ is the Haar probability measure on unit vectors in $\mathcal H_c$.
Minkowski's inequality implies that $h_k(\cdot)^{1/k}$ is subadditive
on positive operators. The Haar representation on the remaining colors gives
\begin{equation}\label{eq:haar-formula}
Z_{k,C}(\tau)
=
\left(\prod_{j\in[q]\setminus\{c\}}
\dim\Sym^k(\mathcal H_j)\right)
\int h_k\!\left(\bra{u_{-c}}\tau\ket{u_{-c}}\right)
\,d\mu_{-c}(u_{-c}),
\end{equation}
where $\mu_{-c}$ is the product of the Haar measures on the color blocks other than $c$.
For each fixed $\ket{u_{-c}}$, set
$A:=\bra{u_{-c}}\tau\ket{u_{-c}}\ge0$.
Although $\Pihi_c$ need not commute with $A$, we have the positive decomposition
$A=A^{1/2}\Pihi_c A^{1/2}+A^{1/2}\Pi^{\mathrm{lo}}_c A^{1/2}$,
since $\Pihi_c+\Pi^{\mathrm{lo}}_c=I$.
Note that the operators $A^{1/2}\Pihi_c A^{1/2}$ and $\Pihi_c A\Pihi_c$ have the same nonzero
eigenvalues; same for $\Pi^{\mathrm{lo}}_c$.
As $h_k$ depends only on the eigenvalues and is unchanged by adjoining
zeros, subadditivity applied to this decomposition yields
\begin{equation}
h_k\!\left(\bra{u_{-c}}\tau\ket{u_{-c}}\right)^{1/k}
\le
h_k\!\left(\bra{u_{-c}}\Pihi_c\tau\Pihi_c
\ket{u_{-c}}\right)^{1/k}
+
h_k\!\left(\bra{u_{-c}}\Pi^{\mathrm{lo}}_c\tau
\Pi^{\mathrm{lo}}_c\ket{u_{-c}}\right)^{1/k}.
\end{equation}

Apply Minkowski's inequality over $\mu_{-c}$ to the functions
$\ket{u_{-c}}\mapsto h_k(\bra{u_{-c}}\Pi\tau\Pi\ket{u_{-c}})^{1/k}$,
for $\Pi\in\{\Pihi_c,\Pi^{\mathrm{lo}}_c\}$. The pointwise bound above
then gives
\begin{equation}
\left(\int h_k\!\left(\bra{u_{-c}}\tau\ket{u_{-c}}\right)\,d\mu_{-c}\right)^{1/k}\le\sum_{\Pi\in\{\Pihi_c,\Pi^{\mathrm{lo}}_c\}}\left(\int h_k\!\left(\bra{u_{-c}}\Pi\tau\Pi\ket{u_{-c}}\right)\,d\mu_{-c}\right)^{1/k}.
\end{equation}
Multiplying by
$\left(\prod_{j\in[q]\setminus\{c\}}\dim\Sym^k(\mathcal H_j)\right)^{1/k}$
and combining with \Cref{eq:haar-formula}, we obtain
\begin{equation}\label{eq:one-color-spectral-step}
Z_{k,C}(\tau)^{1/k}\le Z_{k,C}(\Pihi_c\tau\Pihi_c)^{1/k}+Z_{k,C}(\Pi^{\mathrm{lo}}_c\tau\Pi^{\mathrm{lo}}_c)^{1/k},
\end{equation}
which holds for any positive operator $\tau$ and color $c\in [q]$.

Apply \Cref{eq:one-color-spectral-step} with
$\tau=\omega_{c-1}^{\mathrm{hi}}$. The projectors for different colors
commute, so the two compressed operators are $\omega_c^{\mathrm{hi}}$
and $\omega_c^{\mathrm{lo}}$. Thus for all $c\in [q]$, 
\begin{equation}
Z_{k,C}(\omega_{c-1}^{\mathrm{hi}})^{1/k}
\le Z_{k,C}(\omega_c^{\mathrm{hi}})^{1/k}
+Z_{k,C}(\omega_c^{\mathrm{lo}})^{1/k}.
\end{equation}
Summing over $c\in[q]$ cancels the intermediate high terms and gives
the second inequality in (i). The registerwise accepting projector is
contained in $Q_{k,C}$, so $Z_k(\rho)\le Z_{k,C}(\rho)$, giving the
first inequality.

\noindent\textit{Part (ii).}
Every eigenvalue selected by $\Pihi_c$ is larger than $1/R$, so
$\operatorname{rank}\Pihi_c\le R-1$. Write
$E_c:=\operatorname{ran}\Pihi_c$. The operator
$\omega_q^{\mathrm{hi}}$ is supported on $\bigotimes_{c\in[q]}E_c$.
If some $E_c=\{0\}$, the core is zero and its bound is immediate.
Otherwise, let $\mu_C$ be the product Haar probability measure on unit
vectors $\ket{u_c}\in E_c$, and put
$\ket{u}:=\bigotimes_{c\in[q]}\ket{u_c}$. On this support, each
$Q_{k,c}$ acts as the symmetric projector on $E_c^{\otimes k}$.
Using its integral representation
\cite[Appendix~A, Eqs.~(A3)--(A5)]{weinbrenner2026complete}, we get
\begin{align}
Z_{k,C}(\omega_q^{\mathrm{hi}})
&=\left(\prod_{c\in[q]}\dim\Sym^k(E_c)\right)
\int\!\left(\bra{u}\omega_q^{\mathrm{hi}}\ket{u}\right)^k
\,d\mu_C(u)\notag\\
&=\left(\prod_{c\in[q]}\dim\Sym^k(E_c)\right)
\int\!\left(\bra{u}\rho\ket{u}\right)^k
\,d\mu_C(u)\notag\\
&\le\left(\prod_{c\in[q]}\dim\Sym^k(E_c)\right)
F_C(\rho)^k\notag\\
&\le\bigl(\dim\Sym^k(\C^R)\bigr)^q F_C(\rho)^k.
\end{align}
The second equality uses $\Pihi_c\ket{u_c}=\ket{u_c}$ for every
$c\in[q]$; the inequalities use the definition of $F_C(\rho)$ and
$\dim E_c<R$.

For a discarded part, put
$\sigma_c:=\Tr_{-c}\omega_c^{\mathrm{lo}}$.
Accepting all colors implies accepting color $c$, so
$Q_{k,C}\le Q_{k,c}$. Tracing out the other colors therefore gives
\begin{equation}
\begin{aligned}
Z_{k,C}(\omega_c^{\mathrm{lo}})
&=\Tr\!\left(Q_{k,C}(\omega_c^{\mathrm{lo}})^{\otimes k}\right)\\
&\le\Tr\!\left(Q_{k,c}(\omega_c^{\mathrm{lo}})^{\otimes k}\right)\\
&=\Tr\!\left(Q_{k,c}\sigma_c^{\otimes k}\right)=h_k(\sigma_c).
\end{aligned}
\end{equation}
Here $Q_{k,c}$ performs the $k$-copy permutation test on color block $B_c$, namely
projection onto $\Sym^k(\mathcal H_c)$, and acts as the identity on
the other registers. Since $\sigma_c$ is subnormalized,
$h_k(\sigma_c)$ is the joint probability of obtaining the discarded
part on all $k$ copies and having this test accept. Let $\lambda_i$'s be the eigenvalues of $\sigma_c$, we have
\begin{equation}
h_k(\sigma_c)
=\sum_{1\le i_1\le\cdots\le i_k\le\dim\mathcal H_c}
\lambda_{i_1}\cdots\lambda_{i_k}.
\end{equation}
Thus $h_k$ is the complete homogeneous symmetric polynomial of
degree $k$ in these eigenvalues.

The earlier projections act on other colors; taking the partial trace
therefore gives
\begin{equation}
0\le\sigma_c\le\Pi^{\mathrm{lo}}_c\rho_c\Pi^{\mathrm{lo}}_c
\le\frac{I}{R},
\qquad \Tr\sigma_c\le1.
\end{equation}
Pad the eigenvalues of $\sigma_c$ and a vector with $R$ entries
$1/R$ by zeros to a common length. The former are weakly majorized
by the latter, since the sum of their $j$ largest values is at most
$\min\{j/R,1\}$. The weak-majorization inequality for complete
homogeneous symmetric polynomials
\cite[Corollary~3.3]{issa2021applications} therefore bounds
$h_k(\sigma_c)$ by its value on this vector. The normalization in
that result cancels because the vectors have the same length, giving
\begin{equation}
h_k(\sigma_c)\le\frac{\dim\Sym^k(\C^R)}{R^k}.
\end{equation}
Combining the two bounds and taking $k$-th roots proves (ii).
\end{proof}

\subsection{The Tolerant Testing Algorithm}
\label{subsec:tester-complete}

We now combine \Cref{cor:fq_vs_ffull} and
\Cref{cor:fidelity-sensitive-core} to obtain the tester.
Fix $1\ge a>b\ge0$, $\delta\in(0,1/2)$, $n\ge1$, $d\ge2$, and
$\rho\in\D((\C^d)^{\otimes n})$.
The coloring and spectral projectors are witnesses in the soundness
analysis. The algorithm performs only the local symmetric-subspace
measurements of Protocol~\ref{prot:k-copy-product-test}.

\begin{parameterbox}
Let $\Delta:=a-b$. We split the gap between the errors from random
coloring and blockwise spectral projection by setting
\begin{equation}
  c:=b+\frac{\Delta}{4},
  \qquad
  \eta:=\frac{\Delta}{4}.
  \label{eq:tolerant-c-eta}
\end{equation}
Thus $b<c<a$ and $0<\eta\le c\le1$. Choose
\begin{equation}
  \begin{aligned}
    q&:=\left\lceil\frac{3}{c^2-b^2}\right\rceil,\\
    R&:=\left\lceil\frac{4q}{\eta}\right\rceil,\\
    k&:=\left\lceil
      \frac{8qRc}{\eta}
      \log\!\left(\frac{8eqc}{\eta}\right)
    \right\rceil.
  \end{aligned}
  \label{eq:tester-q-R-k}
\end{equation}
The values of $q$ and $R$ set the test size $k$, but the algorithm
does not construct a coloring or implement blockwise spectral projection. Write
\begin{equation}
  p_{\mathrm Y}:=a^k,
  \qquad
  p_{\mathrm N}:=\left(\frac{a+b}{2}\right)^k,
  \label{eq:tester-one-trial-probabilities}
\end{equation}
so $p_{\mathrm Y}>p_{\mathrm N}$. Set
\begin{equation}
  T:=\left\lceil
    \frac{\min\{2,8p_{\mathrm Y}\}}{(p_{\mathrm Y}-p_{\mathrm N})^2}
    \log\frac{4}{\delta}
  \right\rceil.
  \label{eq:combined-repetitions}
\end{equation}
\end{parameterbox}

\begin{protocol}[Fully tolerant product-state tester]
\label{prot:fully-tolerant-tester}
Given $a,b,\delta$ and independent copies of $\rho$, proceed as follows.
\begin{enumerate}
  \item Set $k$, $T$, $p_{\mathrm Y}$, and $p_{\mathrm N}$ as above.
  \item Run Protocol~\ref{prot:k-copy-product-test} on $T$
  independent groups of $k$ copies, and let $X_j\in\{0,1\}$ indicate
  acceptance in trial $j$.
  \item Accept if $T^{-1}\sum_{j=1}^{T}X_j\ge
  (p_{\mathrm Y}+p_{\mathrm N})/2$; otherwise, reject.
\end{enumerate}
\end{protocol}

The next lemma turns the two bounds from the preceding subsections
into a gap between the one-trial acceptance probabilities.

\begin{lemma}[One-trial acceptance bounds]
\label{lem:tester-one-trial-separation}
For the parameters in \eqref{eq:tolerant-c-eta}--\eqref{eq:tester-one-trial-probabilities}
and any $\rho\in\D((\C^d)^{\otimes n})$, the $k$-copy test satisfies
\begin{equation}
  \begin{aligned}
    \Fprod(\rho)\ge a&\quad\Longrightarrow\quad Z_k(\rho)\ge p_{\mathrm Y},\\
    \Fprod(\rho)\le b&\quad\Longrightarrow\quad Z_k(\rho)\le p_{\mathrm N}.
  \end{aligned}
  \label{eq:tester-one-trial-separation}
\end{equation}
\end{lemma}

\begin{proof}
Put $f:=\Fprod(\rho)$, and choose a pure product vector $\ket{\phi}$
attaining $f$. The accepting projector
$Q_k=\prod_{i=1}^{n}Q_{k,i}$ contains
$\ket{\phi}^{\otimes k}$ in its range. Hence
$Q_k\ge(\ket{\phi}\!\bra{\phi})^{\otimes k}$, and
\begin{equation}
  Z_k(\rho)
  =\Tr(Q_k\rho^{\otimes k})
  \ge\bigl(\bra{\phi}\rho\ket{\phi}\bigr)^k
  =f^k.
  \label{eq:tester-completeness}
\end{equation}
This proves the YES bound, including for mixed inputs.

For soundness, apply Corollary~\ref{cor:fq_vs_ffull} with
$\vepcolor:=c^2-b^2$. We have $0<\vepcolor\le1$, and our choice
of $q$ is exactly $\lceil3/\vepcolor\rceil$. In the NO case,
$f\le b$, so the corollary gives
\begin{equation}
  F_q(\rho)^2
  \le f^2+\vepcolor
  \le b^2+(c^2-b^2)
  =c^2.
  \label{eq:no-Fq-bound}
\end{equation}
Thus $F_q(\rho)\le c$. We next apply
Corollary~\ref{cor:fidelity-sensitive-core} with
$\vep_{\mathrm{sp}}:=\eta$. Since $c\ge\eta$, we have
$\max\{F_q(\rho),\eta\}\le c$, so our choice of $k$ is at least
$k_{\mathrm{sp}}$. Consequently,
\begin{equation}
  Z_k(\rho)^{1/k}
  \le F_q(\rho)+\eta
  \le c+\eta
  =\frac{a+b}{2}.
  \label{eq:tester-soundness}
\end{equation}
Raising both sides to the $k$th power proves the NO bound.
\end{proof}

\begin{theorem}[Fully tolerant product-state testing]
\label{thm:fully-tolerant-tester}
Let $1\ge a>b\ge0$, let $\Delta:=a-b$, and let
$\delta\in(0,1/2)$. For every $n\ge1$ and finite local dimension
$d\ge2$, Protocol~\ref{prot:fully-tolerant-tester}, given independent
copies of any unknown density operator
$\rho\in\D((\C^d)^{\otimes n})$, distinguishes
$\Fprod(\rho)\ge a$ from $\Fprod(\rho)\le b$ with success
probability at least $1-\delta$. It uses
$2^{\widetilde O(\Delta^{-5})}\log(2/\delta)$ copies and
$2^{\widetilde O(\Delta^{-5})}n\,
\poly\!\bigl(\log(2nd/\delta)\bigr)$ quantum gates.
\end{theorem}

\begin{proof}
By Lemma~\ref{lem:tester-one-trial-separation}, each trial has
acceptance probability at least $p_{\mathrm Y}$ in the YES case and
at most $p_{\mathrm N}$ in the NO case. Let
$\overline X:=T^{-1}\sum_{j=1}^{T}X_j$. Stochastic monotonicity
reduces the two error probabilities to Bernoulli means
$p_{\mathrm Y}$ and $p_{\mathrm N}$. At the threshold in
Protocol~\ref{prot:fully-tolerant-tester}, Hoeffding's inequality
bounds each error by $\exp(-T(p_{\mathrm Y}-p_{\mathrm N})^2/2)$.
The multiplicative Chernoff lower-tail bound and Bernstein's
upper-tail bound also bound each error by
$\exp(-T(p_{\mathrm Y}-p_{\mathrm N})^2/(8p_{\mathrm Y}))$;
for the upper tail, use
$p_{\mathrm N}+(p_{\mathrm Y}-p_{\mathrm N})/6\le p_{\mathrm Y}$.
Thus \eqref{eq:combined-repetitions} makes each statistical error at
most $\delta/4$, for a total of at most $\delta/2$.

The exact copy count of this construction is
\begin{equation}
  kT
  =k\left\lceil
    \frac{\min\{2,8a^k\}\log(4/\delta)}
    {\left[a^k-\left((a+b)/2\right)^k\right]^2}
  \right\rceil.
  \label{eq:tester-copy-complexity}
\end{equation}
Since $c-b=\Delta/4$ and $c\ge\Delta/4$, the choice of $q$ gives
\begin{equation}
  q\le1+\frac{48}{\Delta^2},
  \qquad
  qc\le c+\frac{12c}{\Delta(c+b)}
  \le1+\frac{12}{\Delta}.
  \label{eq:tester-q-c-bounds}
\end{equation}
In particular, $qc=O(\Delta^{-1})$. Together with
$R=\lceil16q/\Delta\rceil$, this yields
\begin{equation}
  \begin{aligned}
    q&=O(\Delta^{-2}),&
    qc&=O(\Delta^{-1}),\\
    R&=O(\Delta^{-3}),&
    k&=O\!\left(\Delta^{-5}\log\frac e\Delta\right).
  \end{aligned}
  \label{eq:tester-q-k-asymptotic}
\end{equation}
For the bound on $k$, use $qc/\eta=O(\Delta^{-2})$ in
\eqref{eq:tester-q-R-k}. Moreover,
$p_{\mathrm Y}-p_{\mathrm N}\ge(\Delta/2)a^{k-1}$ and $a\ge\Delta$.
Substituting these estimates into \eqref{eq:tester-copy-complexity}
gives
\begin{equation}
  kT
  \le \exp\!\left(
    O\!\left(\Delta^{-5}\log^2\frac e\Delta\right)
  \right)\log\frac2\delta
  =2^{\widetilde O(\Delta^{-5})}\log\frac2\delta.
\end{equation}

Each trial makes one local symmetric-subspace measurement at each of
the $n$ sites. Implementing each binary measurement to half-diamond-norm
channel error at most $\delta/(2nT)$ changes the final decision
probability by at most $\delta/2$ by a hybrid argument. The total
error is therefore at most $\delta$. The corrected high-dimensional Schur
transform of Burchardt et al.~\cite[Theorem~1]{burchardt2025highdimensional},
followed by measurement of the irrep label $(k)$, implements each local projector using
$\poly(k,\log d,\log(nT/\delta))$ gates per register, giving
\begin{equation}
  nT\,\poly\!\left(k,\log d,\log\frac{nT}{\delta}\right)
  =2^{\widetilde O(\Delta^{-5})}n\,
  \poly\!\left(\log\frac{2nd}{\delta}\right)
\end{equation}
gates in total.
\end{proof}

\section{A Fixed Parameter Tractable Learner}

In this section we are going to give a closest 
    product state learner based on the ideas of 
    our tester above, with fixed parameter tractable (FPT) sample complexity
    and computational complexity.
We have the following result.

\begin{theorem}[An FPT Closest Product State Learner]
    Given $n\ge 1, d\ge 2, \vep,\delta\in \p{0,\f{1}{2}}$, the independent copies of 
    an unknown mixed state $\rho\in \D\p{\p{\C^d}^{\otimes n}}$, 
    there is a quantum algorithm which, with probability at least $1 - \delta$, outputs 
        a classical description of a pure product state $\ket{\wu}$, such that
    \begin{equation}\label{eq:learning-goal}
        \bra{\wu}\rho\ket{\wu} \ge \Fprod\p{\rho} - \vep,
    \end{equation}
    with sample complexity $2^{\widetilde O\p{\vep^{-8}}}(nd)^2
    \poly\p{\log\f{nd}{\delta}}$ and computational complexity $2^{\widetilde O\p{\vep^{-8}}}n^3d^2
    \poly\p{\log\f{nd}{\delta}}$.
\end{theorem}

Without loss of generality, we assume $\Fprod(\rho)\ge\vep$. Otherwise if $\Fprod(\rho)<\vep$, every pure product state trivially fulfills \Cref{eq:learning-goal} above.

\subsection{Overview of the Product State Learning Algorithm}
We use random coloring and blockwise spectral projection to construct a learner for the closest pure product state.
Given the $n$-qudit input state $\rho$, we first randomly choose a coloring $C$ that partitions the qudit set $[n]$ into $q$ parts, where $q$ is independent of $n$.
Let $\ket{x_C} = \bigotimes_{c\in C} \ket{x_c}$ be a closest product state with respect to the partition induced by the coloring $C$.
For the choice of $q$ below, with probability at least $3/4$,
\begin{equation}
  \sum_c\bigl(1-\Fprod(x_c)\bigr)\le\vepcolor=O(\vep).
\end{equation}
This bounds the total internal product infidelity across all blocks.

For each color $c\in C$, given access to the marginal $\rho_c = \Tr_{-c}\rho$, we can then implement the high-eigenvalue subspace projector $\Pihi_c$, whose range, $\ran \Pihi_c = E_c$, is the high-eigenvalue subspace of $\rho_c$.
The dimension of $E_c$ is independent of $n$.
After projecting $\rho$ into $\bigotimes_c E_c$ to obtain $\bar{\rho}$, we then apply Werner's optimal cloning channel \cite{Werner1998} on each $E_c$, where we append several copies of the maximally mixed state over $E_c$ and project them into the symmetric subspace among copies.
For a sufficiently large number of output copies, each local instance of Werner's optimal cloning channel can be approximated by a measure-and-prepare channel. The joint action of these channels can be viewed as drawing ${\ket{v_c}}_{c\in C}$ from a joint distribution and preparing multiple copies of each $\ket{v_c}$. This sampling description is used only in the analysis; the algorithm feeds the clone registers directly to the next step.
In the ideal sampling description, with probability at least
$\exp(-\poly(1/\vep))$, the sampled vectors satisfy
\begin{equation}
  \sum_c\bigl(1-|\langle x_c|v_c\rangle|^2\bigr)=O(\vep).
\end{equation}

Finally, we apply a qudit version of the high-fidelity product-state learner introduced in \cite{bakshi2024learningclosestproductstate} separately to the output registers of the cloning channels action on the color blocks. Combining the resulting classical descriptions yields a pure product-state candidate. We will show that each trial produces a candidate whose fidelity with $\rho$ is sufficiently close to optimal with probability at least $\exp\p{-\poly\p{1/\vep}}$. We repeat the candidate-generation procedure independently, estimate each candidate’s fidelity using fresh copies of $\rho$, and return the candidate with the largest estimated fidelity. With an appropriate number of repetitions and sufficiently accurate fidelity estimates, the resulting learner returns an $\vep$-optimal pure product state with probability at least $1-\delta$.

We use the following qudit version of the high-fidelity product-state learner, which can be easily generalized from the original result in \cite{bakshi2024learningclosestproductstate}. For self-containedness we include a proof in \Cref{app:proof-of-[thm:hfl-qudit]}. 

\begin{restatable}[High-fidelity product state learner, qudit version of Theorem 4.20 of \cite{bakshi2024learningclosestproductstate}]{theorem}{HFL}\label{thm:hfl-qudit}

    Given $n\ge 1, d\ge 2, \vep_{\mathrm{hfl}}\in \left(0, \f{1}{6}\right],\delta\in \p{0,\f{1}{2}}$, and independent copies of
    an unknown mixed state $\rho\in \D\p{\p{\C^d}^{\otimes n}}$, satisfying 
    \begin{equation}
        \Fprod\p{\rho} \ge \f{5}{6} + \vep_{\mathrm{hfl}},
    \end{equation}
    there is a quantum algorithm which, with probability at least $1-\delta$, outputs
    a classical description of a pure product state $\ket{\wu}$, such that
    \begin{equation}
        \bra{\wu}\rho\ket{\wu} \ge \Fprod\p{\rho} - \vep_{\mathrm{hfl}},
    \end{equation}
    with sample complexity $\widetilde{O}\p{\f{nd}{\vep_{\mathrm{hfl}}}}$ and computational complexity $\widetilde{O}\p{\f{\p{nd}^2}{\vep_{\mathrm{hfl}}}}$.
\end{restatable}

\begin{remark}[Relation to the tester]
The learner uses the same random coloring and blockwise spectral projection
ideas as the tester in Section~\ref{sec: tester}. For the tester, the
coloring and spectral projectors serve only as witnesses in the soundness
analysis. For the learner, we sample the coloring and implement blockwise
spectral projection using copies of $\rho$ to prepare the corresponding projected
state. We then implement the local cloning channels to prepare the
registers supplied to the high-fidelity learners, which does not appear in the analysis of the tester. The following
subsections therefore establish both the structural properties and the
implementation guarantees needed for these state preparation operations.
Also note that the required properties also differ subtly: the learner needs nearly product block optimizers for the good random coloring and the blockwise spectral projection that preserves their joint overlap with the state and succeeds with sufficient probability, whereas the tester only needs bounds on optimal fidelity and test acceptance.
\end{remark}

\subsection{Random Coloring}
The first step of the learner is to choose a uniformly random coloring $C: [n] \rightarrow [q]$.
Let $B_c = C^{-1}(c)$ be the qudit block corresponding to the color $c$.
Without ambiguity, we can also use $C$ to represent a partition of $n$ qudits into $q$ parts $\lrs{B_c}_{c\in[q]}$.
Let
\begin{equation}
  F_{C}(\rho)
  =\max_{\ket x=\ket{x_1}\otimes\cdots\otimes\ket{x_q}}
  \bra x\rho\ket x,
  \label{eq:direct-color-optimum}
\end{equation}
and fix an optimizer $\ket{x_C}=\ket{x_1}\otimes\cdots\otimes\ket{x_q}$.
 Define
\begin{equation}
  f_{c}:=\Fprod(x_c)=
  \max_{\ket{u_c}\text{ product over the qudits in }B_c}
  |\langle u_c|x_c\rangle|^2.
  \label{eq:direct-internal-product-fidelity}
\end{equation}
Thus $f_c$ measures the closest product state approximations inside each color block.
For an empty color, set $f_c=1$.

We are going to show that, with high probability, $\ket{x_C}$ and each $\ket{x_c}$ are close to a product state, described by the following lemma.

\begin{lemma}[A good coloring]
\label{cor:direct-good-coloring}
Consider an $n$-qudit state $\rho$ and let $q\ge13$ and $0<f_{\min}\le1$.  Suppose
$F_{C}(\rho)\ge f_{\min}$ for every coloring $C$;
in particular, it is enough that $\Fprod(\rho)\ge f_{\min}$.  Then
\begin{equation}
  \E_C\sum_c(1-f_c^2)
  \le\frac{3}{qf_{\min}^2}.
  \label{eq:direct-expected-internal-defect}
\end{equation}
Moreover, if $\vepcolor>0$ and
\begin{equation}
  q\ge\max\left\{
    13,
    \left\lceil
      \frac{12}{f_{\min}^2\vepcolor}
    \right\rceil
  \right\},
  \label{eq:direct-color-count}
\end{equation}
then, with probability at least $3/4$ over $C$,
\begin{equation}
  \sum_c(1-f_c^2)\le\vepcolor.
  \label{eq:direct-good-color-defect}
\end{equation}
\end{lemma}
The proof of this lemma resembles Lemma \ref{lem:one-block-leakage}, but is somewhat more complicated (because we are demanding something stronger).
\begin{proof}
For each qudit $i\in[n]$, let, as before,  $\swap_i$ be the operator which
    swaps the $i$-th register between two copies of the input.
It is easy to see the operators $\swap_i$ commute pairwise.
For a block $B\subseteq[n]$, note that 
    $\prod_{i\in B}\swap_i$ is the operator which swaps the whole block $B$, we define
\begin{equation}
  Q_{2,B}:=\frac{I+\prod_{i\in B}\swap_i}{2},
  \qquad
  Q^{\mathrm{fine}}_{2,B}:=
      \prod_{i\in B}\frac{I+\swap_i}{2}.
\end{equation}
They are commuting projectors, with
$Q^{\mathrm{fine}}_{2,B}\le Q_{2,B}$. For a coloring $C$, define
\begin{equation}
  \Lambda_{C,c}:=
  \bigl(Q_{2,B_c}-Q^{\mathrm{fine}}_{2,B_c}\bigr)
  \prod_{c'\ne c}Q_{2,B_{c'}}.
\end{equation}
For an empty block both projectors are the identity, so
$\Lambda_{C,c}=0$.
In order to give an upper bound for $\E_C\sum_{c}\p{1-f_c^2}$, the key
    quantity we will look into is $\Tr\p{\Lambda_{C,c}\rho^{\otimes2}}$.

We first prove that, for $q\ge13$,
\begin{equation}
  \mathbb E_C\sum_{c=1}^q\Lambda_{C,c}\le\frac{I}q .
  \label{eq:good-coloring-proof-weighted-bound}
\end{equation}
Since the swap operators $\swap_1,\ldots,\swap_n$ commute, they have simultaneous
eigenspaces indexed by subsets $A\subseteq[n]$, where $\swap_i$ has
eigenvalue $-1$ precisely when $i\in A$. Fix such an eigenspace 
    corresponding to the subset $A$ and
write $N_c:=|A\cap B_c|$. On this eigenspace,
$\sum_c\Lambda_{C,c}$ has eigenvalue
\begin{equation}
  \abs{\{c:N_c>0\}}
  \1\lrb{N_1,\ldots,N_q\text{ are all even}}.
\end{equation}
Indeed, every block must contain an even number of negative swap
signs, and a block contributes to the sum precisely when this number
is nonzero. The eigenvalue vanishes when $A$ is empty or $|A|$ is odd.

Let $h_q(s)$ be the expected value of the eigenvalue 
    when $s=|A|$ is positive and even.
If $s=2$, both sites must have the same color, and hence $h_q(2)=1/q$.
If $s=4$, either all four sites have the same color, or two colors each
contain two sites, giving
\begin{equation}
  h_q(4)=\frac{q+2\cdot 3\cdot q(q-1)}{q^4}
        =\frac{6q-5}{q^3}\le\frac1q
        \qquad(q\ge5).
\end{equation}

For every even $s\ge 6$, it's easy to see that
\begin{equation}
    h_q(s) \le q\Pr_C\lrb{N_1,\ldots,N_q\text{ are all even}}.
\end{equation}
The parity vector $(N_1,\ldots,N_q)\bmod 2$ is distributed as
the position after $s$ steps of the simple random walk on
$\mathbb F_2^q$, started at the origin, which flips one uniformly
chosen coordinate at each step. The formula for its
transition probabilities
\cite[Eqs.~(3.4) and~(4.3)]{slade2023hypercube} therefore gives
\begin{align}
  \Pr_C[N_1,\ldots,N_q \text{ are all even}]
  &=
  2^{-q}\sum_{j=0}^{q}\binom{q}{j}
  \p{1-\frac{2j}{q}}^s \\
  & \le 2^{-q}\sum_{j=0}^{q}\binom{q}{j}
  \p{1-\frac{2j}{q}}^6 \\
  &= \Pr_C[N_1,\ldots,N_q \text{ are all even} \mid s=6].
\end{align}
We further have
\begin{align}
    \Pr_C[N_1,\ldots,N_q \text{ are all even} \mid s=6] 
        &= \f{q + \binom62 q(q-1) + \binom62 \binom42 q(q-1)(q-2)/6}{q^6} \\
    &= \f{15q^2 - 30q + 16}{q^5},
\end{align}
and therefore for every even $s\ge 6$, 
\begin{equation}
    h_q(s) \le q \cdot \f{15q^2 - 30q + 16}{q^5} \le \f{1}{q} \qquad(q\ge 13).
\end{equation}

Recall that we label any joint eigenspace of the swap operators $\swap_1,\dots,\swap_n$ by a subset $A\subseteq [n]$. 
To summarize, for any $A\subseteq [n]$ and the joint eigenspace corresponding to $A$, $\E_C\sum_{c}\Lambda_{C,c}$ has eigenvalue bounded by $\f{1}{q}$. 
We conclude that \Cref{eq:good-coloring-proof-weighted-bound} holds.

We next show that, for every coloring $C$,
\begin{equation}
  \sum_{c}\Tr\left(\Lambda_{C,c}\rho^{\otimes2}\right)
  \ge \f{f_{\min}^2}{3} \sum_{c} \p{1-f_c^2}.
  \label{eq:good-coloring-proof-block-leakage}
\end{equation}

For a nonempty color block, apply the one-color refinement estimate
\Cref{eq:one-color-refinement-leakage} from the proof of
Lemma~\ref{lem:one-block-leakage} to the fixed optimizer $\ket{x_C}$.
Since $\Lambda_{C,c}=Q_{2,C}-Q_{2,C^{(c)}}$ and
$f_c=\Fprod(x_c)$, this gives, for all $c\in [q]$,
\begin{equation}
  \Tr\!\left(\Lambda_{C,c}\rho^{\otimes2}\right)
  \ge
  \frac{F_{C}(\rho)^2}{3}(1-f_c^2).
\end{equation}
Summing over $c$, taking the coloring expectation, and using
$F_{C}(\rho)\ge f_{\min}$ gives
\begin{equation}
  \frac{f_{\min}^2}{3}\,
  \mathbb E_C\sum_c(1-f_c^2)
  \le
  \mathbb E_C\sum_c
  \Tr\!\left(\Lambda_{C,c}\rho^{\otimes2}\right)
  \le\frac1q,
\end{equation}
which proves \Cref{eq:direct-expected-internal-defect}.  Markov's
inequality then gives
\begin{equation}
  \Pr_C\left[
    \sum_c(1-f_c^2)>\vepcolor
  \right]
  \le
  \frac{3}{qf_{\min}^2\vepcolor}
  \le\frac14,
\end{equation}
where the last inequality is \Cref{eq:direct-color-count}.  This proves
\Cref{eq:direct-good-color-defect}.
\end{proof}

Let $\rho_c:=\Tr_{-B_c}\rho$ be the marginal of $\rho$ on the color
block $B_c$. The following property of these marginals will be used
in the next blockwise spectral projection step.

\begin{lemma}[Marginal dominance]
\label{lem:learner-marginal-dominance}
Fix a coloring $C$ and an optimizer
$\ket{x_C}=\bigotimes_c\ket{x_c}$ in
\eqref{eq:direct-color-optimum}. Then, for every color $c$,
\begin{equation}
  \rho_c\ge F_C(\rho)\ket{x_c}\bra{x_c}.
  \label{eq:learner-marginal-dominance}
\end{equation}
\end{lemma}

\begin{proof}
Fix a color $c$. Let $\ket{x_{-c}}:=\bigotimes_{c'\ne c}\ket{x_{c'}}$
and define
\begin{equation}
  T_c:=(I_{B_c}\otimes\bra{x_{-c}})\rho
       (I_{B_c}\otimes\ket{x_{-c}}).
\end{equation}
It is easy to see that $\rho_c=\Tr_{-B_c} \rho\ge T_c$.
Optimality gives $\bra vT_c\ket v\le F_C(\rho)$ for every unit vector
$\ket v\in\mathcal H_{c}$, with equality at $\ket{x_c}$.
Thus $\ket{x_c}$ is a top eigenvector of the positive operator $T_c$,
and $T_c\ge F_C(\rho)\ket{x_c}\bra{x_c}$.
We conclude that $\rho_c \ge F_C(\rho)\ket{x_c}\bra{x_c}$.
\end{proof}

\begin{parameterbox}
Now we are going to explicitly state our choice of parameters which appear in this section.
These parameters will lead to the desired learning result, \Cref{thm:ProductLearner}.
\begin{equation}
    \vep_0:=\min\{\vep,1/10\},\qquad
    f_{\min}:=\frac{\vep_0}{8},\qquad
    \vepcolor:=\frac{\vep_0}{10^4},
\end{equation}
and choose
\begin{equation}
    q:=\max\left\{13,
        \left\lceil\frac{12}{f_{\min}^2\vepcolor}\right\rceil
    \right\}
    =O(\vep^{-3}).
\end{equation}
When $\Fprod(\rho)\ge f_{\min}$,
Lemma~\ref{cor:direct-good-coloring} gives
\begin{equation}
    \Pr_C\left[\sum_c(1-f_c^2)\le\vepcolor\right]
    \ge\frac34.
\end{equation}

We note that if $\Fprod(\rho) < f_{\min} < \vep$, the learning problem becomes trivial and any pure product state output will be $\vep$-close to the input state.
\end{parameterbox}

\subsection{Blockwise Spectral Projection}\label{sec:blockwise-spectral-projection}
After choosing a coloring $C$ for the $n$ qudits, the next step of the learning algorithm is to implement blockwise spectral projections onto the high eigenvalue subspaces of the local marginals. .
Define $\Pihi_C = \bigotimes_{c}\Pihi_c$, where 
$\Pihi_c = \1_{\lrb{\theta_c,1}}\p{\rho_c}$.
Let $E_c = \ran \Pihi_c$.
Then we apply this projector to the input state, 
\begin{equation}
    \rhobar_C = \f{\Pihi_C\rho\Pihi_C}{\Tr\lrb{\Pihi_C\rho}}.
\end{equation}
For analysis, we define
\begin{equation}
  \ket{\xbar_c}:=
  \frac{\Pihi_c\ket{x_c}}{\norm{\Pihi_c\ket{x_c}}_2},
  \qquad
  \ket{\xbar_C}=\bigotimes_c\ket{\xbar_c}
\end{equation}

We are going to show the following:

\begin{enumerate}
    \item The dimension of each $E_c$ is bounded, independent of $n$.
    \item When we apply $\Pihi_C$ to $\rho$, the projection succeeds with probability at least $\pproj$, where $\pproj$ depends only on $\vep$.
    \item After projection, the projected input state has overlap at least $\pproj$ with the projected optimizer; that is, $\bra{\xbar_C}\rhobar_C\ket{\xbar_C} \ge \pproj$.
    \item $\Pihi$ is easy to implement. We can efficiently use copies of $\rho_c$ to implement $\Pihi_c$, with small errors.
\end{enumerate}

Each property is crucial for the next cloning step.
After this projection we restrict our state in the small subspace $\bigotimes_cE_c$, which guarantees that the next step of the algorithm (the cloning step) can be implemented efficiently. In the meantime, the overlap between the projected state $\rhobar$ and the projected optimizer $\ket{\xbar}$ is large enough, to make sure we do not lose much information on the closest product state. 

\subsubsection{Properties of the Blockwise Spectral Projection $\Pihi_C$}
\label{subsubsec:learner-projection-properties}

Fix $\vepproj\in(0,1/4)$ and set
\begin{equation}
  \theta_0:=\frac{f_{\min}\vepproj}{2q}.
  \label{eq:learner-spectral-parameters}
\end{equation}
We allow each threshold to be chosen anywhere in
$\theta_c\in[\theta_0,2\theta_0]$. This freedom will be useful when
implementing the projectors in the next subsection.
For an empty color block, the Hilbert space is one-dimensional and
$\Pihi_c=I$.

\begin{lemma}[Properties of blockwise spectral projection]
\label{lem:learner-spectral-projection}
For every choice of thresholds $\theta_c\in[\theta_0,2\theta_0]$,
the blockwise spectral projector $\Pihi_C = \bigotimes_{c}\Pihi_c$ defined above satisfies
\begin{align}
  \dim E_c
  &\le \frac1{\theta_0},
  \label{eq:learner-spectral-rank}\\
  \sum_c\norm{(I-\Pihi_c)\ket{x_c}}_2^2
  &\le\vepproj.
  \label{eq:learner-spectral-capture}
\end{align}
In particular, the normalized vectors
$\ket{\xbar_c}$ and $\ket{\xbar_C}$ are well defined and satisfy
\begin{align}
    \sum_c\left(1-\abs{\braket{x_c}{\xbar_c}}^2\right)
    &=\sum_c\norm{(I-\Pihi_c)\ket{x_c}}_2^2
    \le\vepproj,\\
    \abs{\braket{x_C}{\xbar_C}}^2
    &\ge1-\vepproj.
\end{align}
Furthermore, the projection of $\Pihi_C$ on $\rho$ succeeds with probability at least
$\pproj$, and the normalized projected state has overlap at least
$\pproj$ with the projected optimizer:
\begin{equation}
  \Tr(\Pihi_C\rho)\ge\pproj,
  \qquad
  \bra{\xbar_C}\rhobar_C\ket{\xbar_C}\ge\pproj,
  \label{eq:learner-projection-success-overlap}
\end{equation}
where $\pproj:=f_{\min}(1-\vepproj)^2$.
\end{lemma}

\begin{proof}
Every eigenvalue of $\rho_c$ on $E_c$ is at least
$\theta_c\ge\theta_0$. Since $\Tr(\rho_c)=1$, we have
\begin{equation}
  \theta_0\dim E_c \le \theta_c\dim E_c \le \Tr(\Pihi_c\rho_c) \le 1,
\end{equation}
which proves \eqref{eq:learner-spectral-rank}.

By Lemma~\ref{lem:learner-marginal-dominance} and the
definition of the spectral projector, for any $c$,
\begin{equation}
    (I-\Pihi_c)\ketbra{x_c}{x_c}(I-\Pihi_c)
        \le \f{1}{F_C(\rho)} (I-\Pihi_c)\rho_c(I-\Pihi_c)
        \le \f{\theta_c}{F_C(\rho)} I,
\end{equation}
therefore we have
\begin{align}
    \norm{(I-\Pihi_c)\ket{x_c}}_2^4
        &= \bra{x_c}(I-\Pihi_c)\ketbra{x_c}{x_c}(I-\Pihi_c)\ket{x_c} \\
        &\le \f{\theta_c}{F_C(\rho)} \bra{x_c}(I-\Pihi_c) \cdot I \cdot (I-\Pihi_c) \ket{x_c} \\
        &\le \f{\theta_c}{F_C(\rho)} \norm{(I-\Pihi_c)\ket{x_c}}_2^2.
\end{align}
We conclude that 
\begin{equation}\label{eq:learner:bound on difference}
    \norm{(I-\Pihi_c)\ket{x_c}}_2^2 \le \f{\theta_c}{F_C(\rho)},
\end{equation}
and
\begin{equation}
  \sum_c \norm{(I-\Pihi_c)\ket{x_c}}_2^2 
  \le \frac{\sum_c\theta_c}{F_C(\rho)}
  \le \frac{2q\theta_0}{F_C(\rho)}
  =\frac{f_{\min}}{F_C(\rho)}\vepproj
  \le\vepproj,
\end{equation}
proving \eqref{eq:learner-spectral-capture}.
In particular, $\norm{\Pihi_c\ket{x_c}}_2 \ge 1 - \vepproj > 0$ for every $c$, so all the normalized vectors $\ket{\xbar_c}$ are well defined. Moreover,
\begin{equation}
  \abs{\braket{x_c}{\xbar_c}}^2
      = \bra{x_c}\f{\Pihi_c\ketbra{x_c}{x_c}\Pihi_c}{\norm{\Pihi_c\ket{x_c}}_2^2}\ket{x_c}
      = \bra{x_c}\Pihi_c\ket{x_c}
      = 1 - \norm{(I-\Pihi_c)\ket{x_c}}_2^2,
\end{equation}
and hence
\begin{equation}
  \sum_c\left(1-\abs{\braket{x_c}{\xbar_c}}^2\right)
  =\sum_c \norm{(I-\Pihi_c)\ket{x_c}}_2^2\le\vepproj,
\end{equation}
and
\begin{equation}
  \abs{\braket{x_C}{\xbar_C}}^2
  =\prod_c\p{1-\norm{(I-\Pihi_c)\ket{x_c}}_2^2}
  \ge 1-\sum_c \norm{(I-\Pihi_c)\ket{x_c}}_2^2
  \ge 1-\vepproj.
\end{equation}

It remains to bound the projection success probability
and the overlap after projection. 
First we are going to show $\bra{x_C}\Pihi_C\rho\Pihi_C \ket{x_C} = \norm{\sqrt{\rho}\Pihi_C\ket{x_C}}_2^2 \ge \pproj$.

For each $c$, define
\begin{equation}
  Z_c := \bigotimes_{j<c} I_j  \otimes (I-\Pihi_c) \otimes \bigotimes_{j>c}\Pihi_j .
\end{equation}
We have
\begin{equation}
    I - \Pihi_C = \sum_{c=1}^q Z_c.
\end{equation}

Note that 
\begin{equation}
    Z_c\ketbra{x_C}{x_C}Z_c \le \bigotimes_{j<c}I_j \otimes (I-\Pihi_c)\ketbra{x_c}{x_c}(I-\Pihi_c) \otimes \bigotimes_{j>c}I_j.
\end{equation}
By the definition of $\rho_c$, we have
\begin{align}
    \norm{\sqrt{\rho}Z_c\ket{x_C}}_2^2
        &= \bra{x_C}Z_c\rho Z_c\ket{x_C} \\
        &= \Tr\lrb{\rho Z_c\ketbra{x_C}{x_C}Z_c} \\
        &\le \Tr\lrb{\rho_c (I-\Pihi_c) \ketbra{x_c}{x_c} (I-\Pihi_c)} \\
        &= \bra{x_c}(I-\Pihi_c)\rho_c(I-\Pihi_c)\ket{x_c} \\
        &\le \theta_c\norm{(I-\Pihi_c)\ket{x_c}}_2^2 \\
        &\le \frac{\theta_c^2}{F_C(\rho)},
\end{align}
where the last inequality follows from the bound
proved above, \Cref{eq:learner:bound on difference}. Applying the triangle inequality gives
\begin{align}
    \norm{\sqrt{\rho}(I-\Pihi_C)\ket{x_C}}_2
        &\le \sum_{c=1}^q
            \norm{\sqrt{\rho}Z_c\ket{x_C}}_2 \\
        &\le \frac{\sum_c\theta_c}{\sqrt{F_C(\rho)}} \\
        &\le \frac{2q\theta_0}{\sqrt{F_C(\rho)}} \\
        &= \frac{f_{\min}\vepproj}{\sqrt{F_C(\rho)}} \\
        &\le \sqrt{F_C(\rho)}\,\vepproj.
\end{align}
Since $\norm{\sqrt{\rho}\ket{x_C}}_2=\sqrt{F_C(\rho)}$,
the triangle inequality implies
\begin{align}
    \norm{\sqrt{\rho}\Pihi_C\ket{x_C}}_2
        &\ge \norm{\sqrt{\rho}\ket{x_C}}_2
            -\norm{\sqrt{\rho}(I-\Pihi_C)\ket{x_C}}_2 \\
        &\ge \sqrt{F_C(\rho)}(1-\vepproj).
\end{align}
Consequently,
\begin{equation}
    \bra{x_C}\Pihi_C\rho\Pihi_C\ket{x_C}
        \ge F_C(\rho)(1-\vepproj)^2
        \ge f_{\min}(1-\vepproj)^2
        =\pproj.
\end{equation}

Because $\Pihi_C\rho\Pihi_C$ is positive semidefinite,
we obtain
\begin{equation}
    \Tr(\Pihi_C\rho)
        =\Tr(\Pihi_C\rho\Pihi_C)
        \ge \bra{x_C}\Pihi_C\rho\Pihi_C\ket{x_C}
        \ge\pproj>0.
\end{equation}
Thus the projection succeeds with probability at least
$\pproj$, and $\rhobar_C$ is well defined.

Finally, using
\begin{equation}
    \ket{\xbar_C}
        =\frac{\Pihi_C\ket{x_C}}
            {\norm{\Pihi_C\ket{x_C}}_2},
\end{equation}
we have
\begin{align}
    \bra{\xbar_C}\rhobar_C\ket{\xbar_C}
        &=
        \frac{\bra{x_C}\Pihi_C\rho\Pihi_C\ket{x_C}}
        {\norm{\Pihi_C\ket{x_C}}_2^2
            \Tr(\Pihi_C\rho)} \\
        &\ge \bra{x_C}\Pihi_C\rho\Pihi_C\ket{x_C} \\
        &\ge\pproj,
\end{align}
where the first inequality holds because both factors
in the denominator are positive and at most one.
This proves \eqref{eq:learner-projection-success-overlap}.
\end{proof}

\subsubsection{Implementing the Blockwise Spectral Projection $\Pihi_C$}
\label{subsubsec:learner-projection-implementation}

We implement the projectors $\Pihi_c$ from
Lemma~\ref{lem:learner-spectral-projection} using independent copies of
the marginals $\rho_c$, obtained by discarding all color blocks other
than $c$ from fresh copies of $\rho$. We construct a spectral
threshold filter and use state samples to implement the corresponding
two-outcome measurement, with outcome $1$ for the high-eigenvalue subspace.
For an $m$-qubit operator $A$, an exact block encoding
with $a$ ancillas is a unitary $U$ satisfying
\begin{equation}
  (\bra{0^a}\otimes I)U(\ket{0^a}\otimes I)=A.
\end{equation}
To implement the projectors $\Pihi_c$, we use two complementary tools.
Theorem~\ref{thm:learner-qsvt-threshold} constructs an approximate spectral
filter for any quantum state that preserves high-eigenvalue components and suppresses
low-eigenvalue components, assuming access to a block encoding of the input state.
Theorem~\ref{thm:learner-samplizer} allows us to simulate such circuits
using copies of the input state instead of accessing to a block encoding of the state.
Together, these results enable us to approximate the required blockwise spectral projections with copies of $\rho_c$.

\begin{theorem}[Spectral threshold filtering, rephrased from
  {\cite[Theorem~19]{GSLW19}}]
\label{thm:learner-qsvt-threshold}
Let $A=\sum_j\lambda_j\ket{v_j}\!\bra{v_j}$ satisfy $0\le A\le I$
on $m$ qubits. For any exact $a$-ancilla block encoding $U$ of $A$
and known computable parameters $\theta,\Delta_{\mathrm{sp}}\in(0,1)$ and
$\varepsilon\in(0,1/2]$, we can implement $W[U]$ and $W[U]^\dagger$, such that
\[
  F:=\p{\bra{0^{a+1}}\otimes I} W[U] \p{\ket{0^{a+1}}\otimes I}
    =\sum_j f_j\ket{v_j}\!\bra{v_j},
  \qquad
  \begin{cases}
    f_j\in[0,\varepsilon],&\lambda_j\le\theta-\Delta_{\mathrm{sp}},\\
    f_j\in[1-\varepsilon,1],&\lambda_j\ge\theta+\Delta_{\mathrm{sp}}.
  \end{cases}
\]
The circuit uses $O(\Delta_{\mathrm{sp}}^{-1}\log(1/\varepsilon))$ queries to
$U,U^\dagger$ and $O((a+1)\Delta_{\mathrm{sp}}^{-1}\log(1/\varepsilon))$
additional one- and two-qubit gates, with $O(a+1)$ ancillae. The non-oracle gates are independent of $U$.
\end{theorem}

\begin{theorem}[Sample-based circuit simulation (Samplizer),
  {\cite[Theorem III.1]{WangZhang2025Samplizer}}]
\label{thm:learner-samplizer}
Let $\omega$ be an $m$-qubit density operator and let $a\ge4$.
Consider a circuit family $C[U]$ with $Q\ge1$ queries to an exact
$a$-ancilla block encoding $U$ of $\omega/2$ or its inverse, and $T$ additional one- and two-qubit gates independent of $U$.
For every $0<\xi\le1/2$, there is a channel
$\mathcal S_{\omega,\xi}$ using
\begin{equation}
  O\!\left(\frac{Q^2}{\xi}
       \log^2\frac{Q}{\xi}\right)
\end{equation}
independent copies of $\omega$ and $O\p{T+m\frac{Q^2}{\xi}
       \log^2\frac{Q}{\xi}}$ one- and two-qubit gates.
For each $\omega$, there exists an exact $a$-ancilla block encoding $U_\omega$
of $\omega/2$ for which
\begin{equation}
  \bigl\|\mathcal S_{\omega,\xi}
       -\mathcal C[U_\omega]\bigr\|_\diamond\le\xi,
\end{equation}
where the channel $\mathcal C[U_\omega]$ is defined as $\mathcal C[U_\omega](X) = C[U_\omega]XC[U_\omega]^\dagger$.
\end{theorem}

Note that the first theorem works for every block encoding, so it applies to the particular one denoted by the second. 

\begin{lemma}[Sample-based threshold spectral measurement implementation]
\label{lem:learner-guarded-filter}
Let $\omega$ be a density operator on $m$ qudits of local dimension
$d\ge2$, and let known computable $\theta\in(0,1)$ and
$\Delta_{\mathrm{sp}}\in(0,1]$
satisfy the gap condition
\begin{equation}
  \operatorname{spec}(\omega)\cap
  (\theta-\Delta_{\mathrm{sp}},\theta+\Delta_{\mathrm{sp}})
  =\varnothing.
  \label{eq:learner-guarded-gap}
\end{equation}
Set $\Pihi=\1_{[\theta,1]}(\omega)$, and define the measurement channel
\begin{equation}
  \mathcal M_\theta(X)
  =\Pihi X\Pihi\otimes\ket{1}\!\bra{1}
   +(I-\Pihi)X(I-\Pihi)\otimes\ket{0}\!\bra{0}.
\end{equation}
For every $0<\xi_{\mathrm{proj}}\le1/2$, we can implement a channel
$\widetilde{\mathcal M}_\theta$ satisfying
\begin{equation}
  \bigl\|\widetilde{\mathcal M}_\theta
    -\mathcal M_\theta\bigr\|_\diamond\le\xi_{\mathrm{proj}}
\end{equation}
using $\widetilde O(\Delta_{\mathrm{sp}}^{-2}\xi_{\mathrm{proj}}^{-1})$
independent copies of $\omega$ and
$\widetilde O((m\lceil\log_2d\rceil+1)\Delta_{\mathrm{sp}}^{-2}\xi_{\mathrm{proj}}^{-1})$
one- and two-qubit gates under a binary encoding of the qudits.
\end{lemma}

\begin{proof}
If $\Delta_{\mathrm{sp}}>\theta$, then $\Pihi=I$ and the claim is
immediate, so assume $\Delta_{\mathrm{sp}}\le\theta$.
First consider the qubit case $d=2$.
Apply \Cref{thm:learner-qsvt-threshold} to $A=\omega/2$ with
threshold $\theta/2$, transition half-width $\Delta_{\mathrm{sp}}/2$,
$a=4$, and $\varepsilon=\xi_{\mathrm{proj}}/(16\sqrt{2})$.
Append five zero ancillas and a zero outcome qubit. Apply $W[U]$,
flip the outcome qubit conditioned on the five ancillas being all
zero, and apply $W[U]^\dagger$.

The spectral bounds of $F$ give $\norm{F^2-\Pihi}_\infty\le2\varepsilon$,
now we are going to derive a diamond norm bound based on this infinity norm bound by applying the block-measurement theorem of
Rall~\cite[Theorem~19]{Rall2021BlockMeasurement} to the block encoding
$W[U]$ of $F$, which shows that the diamond norm could be bounded by the infinity norm bound with a multiplicative factor $4\sqrt{2}$.
After discarding the ancilla register and measuring the
outcome qubit, the resulting channel $\widetilde{\mathcal M}_\theta$
satisfies
\begin{equation}
  \|\widetilde{\mathcal M}_\theta-\mathcal M_\theta\|_\diamond
  \le4\sqrt{2}\|F^2-\Pihi\|_\infty
  \le8\sqrt{2}\varepsilon=\xi_{\mathrm{proj}}/2,
\end{equation}
where measuring the outcome qubit cannot increase the diamond-norm
error guaranteed by the theorem.

The entire unitary circuit uses
$O(\Delta_{\mathrm{sp}}^{-1}\log(1/\xi_{\mathrm{proj}}))$ queries and $O(\Delta_{\mathrm{sp}}^{-1}\log(1/\xi_{\mathrm{proj}}))$
additional gates. Apply Theorem~\ref{thm:learner-samplizer} once to
this circuit with simulation error $\xi_{\mathrm{proj}}/2$, obtaining the claimed
sample and gate bounds for $d=2$; contractivity under initialization,
discarding ancillas, and measurement, followed by the triangle
inequality, gives total diamond-norm error at most $\xi_{\mathrm{proj}}$.

For general $d$, encode each qudit in $\lceil\log_2d\rceil$ qubits
and apply the qubit case with $m'=m\lceil\log_2d\rceil$, then decode
the output.
\end{proof}

The following elementary counting argument supplies a common eigenvalue gap $\Delta_{\mathrm{sp}}$ for all colors, using the threshold scale $\theta_0$ from
\eqref{eq:learner-spectral-parameters}.

\begin{lemma}
\label{lem:learner-random-guard}
Let $0<\theta_0<1/2$, $q\ge1$, and
$0<\delta_{\mathrm{gap}}<1$. Choose each $\theta_c$ independently
and uniformly from $[\theta_0,2\theta_0]$, and set
\begin{equation}
  \Delta_{\mathrm{sp}}
  :=\frac{\delta_{\mathrm{gap}}\theta_0^2}{4q}.
  \label{eq:learner-spectral-guard-width}
\end{equation}
With probability at least $1-\delta_{\mathrm{gap}}$, every marginal
$\rho_c$ satisfies \eqref{eq:learner-guarded-gap} at its threshold
$\theta_c$.
The same conclusion holds if $N$ is an integer with
$h:=\theta_0/N\le\Delta_{\mathrm{sp}}$ and each $\theta_c$ is chosen
independently and uniformly from
$\{\theta_0+jh:j=0,\ldots,N-1\}$.
\end{lemma}

\begin{proof}
Since $\Delta_{\mathrm{sp}}<\theta_0/2$, only eigenvalues in
$[\theta_0/2,5\theta_0/2]$ can lie within $\Delta_{\mathrm{sp}}$ of
a chosen threshold. A density operator has at most $2/\theta_0$
such eigenvalues, counting multiplicity. Each excludes an interval
of length $2\Delta_{\mathrm{sp}}$. Consequently, for one color,
\begin{equation}
  \Pr\!\left[
    \operatorname{spec}(\rho_c)\cap
    (\theta_c-\Delta_{\mathrm{sp}},\theta_c+\Delta_{\mathrm{sp}})
    \ne\varnothing\right]
  \le \frac{2}{\theta_0}\cdot 2 \Delta_{\mathrm{sp}} \cdot \frac{1}{\theta_0}
  =\frac{\delta_{\mathrm{gap}}}{q}.
\end{equation}
A union bound over the $q$ colors proves the claim.
For the finite grid, only eigenvalues above
$\theta_0-\Delta_{\mathrm{sp}}>7\theta_0/8$ matter, so there are at
most $8/(7\theta_0)$ of them. Each excludes at most
$2\Delta_{\mathrm{sp}}/h+1$ grid points, a fraction at most
$(2\Delta_{\mathrm{sp}}+h)/\theta_0\le3\Delta_{\mathrm{sp}}/\theta_0$.
A union bound over all colors gives joint failure probability at most
$6\delta_{\mathrm{gap}}/7$.
\end{proof}

\begin{parameterbox}
Now we are going to explicitly state our choice of parameters which appear in this section.
These parameters will lead to the desired learning result, \Cref{thm:ProductLearner}.
Recall the parameters from the random-coloring step:
\begin{equation}
    \vep_0:=\min\{\vep,1/10\},
    \qquad
    f_{\min}:=\frac{\vep_0}{8},
    \qquad
    \vepcolor:=\frac{\vep_0}{10^4},
\end{equation}
and
\begin{equation}
    q:=\max\left\{13,
        \left\lceil\frac{12}{f_{\min}^2\vepcolor}\right\rceil
    \right\}
    =\Theta(\vep^{-3}).
\end{equation}
For the blockwise spectral projection step, set
\begin{equation}
    \vepproj:=\frac{\vep_0}{10^4},
    \qquad
    \theta_0:=\frac{f_{\min}\vepproj}{2q},
    \qquad
    \delta_{\mathrm{gap}}:=\frac18.
\end{equation}
We have $\theta_0=\Theta(\vep^5)$. Thus,
Lemma~\ref{lem:learner-spectral-projection} gives
\begin{equation}
    \dim E_c\le\theta_0^{-1}=O(\vep^{-5}),
    \qquad
    \pproj=f_{\min}(1-\vepproj)^2=\Theta(\vep),
\end{equation}
while the total squared projection loss of the optimizer
is at most $\vepproj$.
Choose
\begin{equation}
  \begin{aligned}
    \Delta_{\mathrm{sp}}
    &:=\frac{\delta_{\mathrm{gap}}\theta_0^2}{4q}
      =\Theta(\vep^{13}),\\
    N&:=2^{\left\lceil
        \log_2(\theta_0/\Delta_{\mathrm{sp}})
    \right\rceil},
    \qquad h:=\frac{\theta_0}{N}.
  \end{aligned}
\end{equation}
Independently for every color, choose $J_c$ uniformly from
$\{0,\ldots,N-1\}$ and set
$\theta_c:=\theta_0+J_ch$.
By the finite-grid conclusion of
Lemma~\ref{lem:learner-random-guard}, all thresholds satisfy the
required gap conditions simultaneously with probability at least $7/8$.

Conditioned on this event, for a target implementation
accuracy $\eta\in(0,1/2]$, implement each local measurement
using Lemma~\ref{lem:learner-guarded-filter} with
\begin{equation}
    \xi:=\frac{\eta\pproj}{4q}.
\end{equation}
The diamond-norm error of implementing the joint measurement $\Pihi_C = \bigotimes_c \Pihi_c$ is bounded by $q\xi=\eta\pproj/4$. 
Consequently, the implemented projection succeeds with 
probability at least $\pproj/2$, and its normalized 
accepted output differs from $\rhobar_C$ by at most 
$\eta$ in trace norm.
We choose $\eta$ in the final analysis relative to the
success probability of the subsequent cloning step.
\end{parameterbox}

\subsection{Werner's Optimal Cloning Channel}
\label{subsec:learner-werner-cloning}
After projecting $\rho$ with $\Pihi_C$, we restrict the input state to the bounded-dimensional
space $\bigotimes_c E_c$, while preserving overlap with the projected
optimizer $\ket{\xbar_C}$. This does not yet allow us to apply the
high-fidelity learner, because the overlap guarantee
$\bra{\xbar_C}\rhobar_C\ket{\xbar_C}\ge\pproj$
does not imply that the product fidelity of $\rhobar_C$ for each color block exceeds $5/6 + \vep_{\mathrm{hfl}}$.

In this section, we use Werner's optimal cloning channel to solve this problem. 
Write $d_c:=\dim E_c$, and let $\Pi_{L,E_c}$ be the projector onto
$\Sym^L(E_c)$. For an integer $L\ge1$, Werner's optimal $1\to L$ cloning channel on
$E_c$ is the channel~\cite{Werner1998}
\begin{equation}
  \mathcal W_{c,L}(X)
  :=\frac{d_c}{\binom{L+d_c-1}{L}}
    \Pi_{L,E_c}
    \left(X\otimes I_{E_c}^{\otimes(L-1)}\right)
    \Pi_{L,E_c}.
  \label{eq:learner-werner-channel}
\end{equation}
Its input is supported on $E_c$. Applying these channels to the
projected input gives the cloned state
\begin{equation}
  \Omega_{C,L}
  :=\left(\bigotimes_c\mathcal W_{c,L}\right)(\rhobar_C).
  \label{eq:learner-joint-clone-output}
\end{equation}

\subsubsection{Properties of the Cloned State $\Omega_{C,L}$}

We will need two properties for learning.
First, for sufficiently many output registers, each local instance of Werner's optimal cloning channel approximates a measure-and-prepare channel: it draws one unit vector $\ket{v_c}\in E_c$ for each color and prepares multiple copies of that same vector.
We can further conclude that the output state $\Omega_{C,L}$ the simultaneous action of all cloning channels  can be analyzed as repeated copies of a pure product state across colors, where the state is sampled from some joint distribution.
Second, with probability at least $\exp\p{-\poly\p{1/\vep}}$, the sampled product state is close enough to the projected optimizer $\ket{\bar{x}_C}$ on each color block, to make sure the product fidelity of the sampled state is over $5/6 + \vep_{\mathrm{hfl}}$.
This allows us to feed $\Omega_{C,L}$ to the high-fidelity closest product state learner.
We will prove these two properties in the following lemma.

\begin{lemma}[Properties of the cloned state $\Omega_{C,L}$]
\label{lem:learner-werner-cloning}
Let $\rhobar_C$, $E_c$, and $\ket{\xbar_C}$ be as in
Lemma~\ref{lem:learner-spectral-projection}, and write $d_c=\dim E_c$.
There is a probability distribution $\nu_C$ on the unit vectors in $\bigotimes_cE_c$, independent of $L$, with the following properties.
\begin{enumerate}[label=(\roman*)]
  \item For every integer $L\ge1$,
  \begin{equation}
    \left\|
      \Omega_{C,L}
      -\mathbb E_{\nu_C}\left[
        \bigotimes_c\left(\ket{v_c}\!\bra{v_c}\right)^{\otimes L}
      \right]
    \right\|_1
    \le\sum_c\frac{2d_c}{L+d_c}
    \le\frac{2q}{\theta_0L}.
    \label{eq:learner-cloning-comparison}
  \end{equation}
  \item For every $\vepclone\in(0,1)$,
  \begin{equation}
    \Pr_{\nu_C}\!\left[
      \sum_c\left(1-\abs{\braket{\xbar_c}{v_c}}^2\right)
      \le\vepclone
    \right]
    \ge
    \pproj\left(\frac{\vepclone}{q}\right)^{
      \sum_c(d_c-1)}
    \ge \pclone,
    \label{eq:learner-cloning-good-event}
  \end{equation}
  where $\pclone :=  \pproj\left(\frac{\vepclone}{q}\right)^{
    q(\lceil\theta_0^{-1}\rceil-1)}$.
\end{enumerate}
\end{lemma}

\begin{proof}
Let $\mu_c$ be the Haar measure on the unit vectors in $E_c$, we define the joint distribution $\nu_C$ on the unit vectors in $\bigotimes_cE_c$ as follows,
\begin{equation}
  \mathrm d\nu_C(v_1,\ldots,v_q)
  :=\left(\prod_c d_c\right)
    \bra{v_1\otimes\cdots\otimes v_q}
    \rhobar_C
    \ket{v_1\otimes\cdots\otimes v_q}
    \prod_c\mathrm d\mu_c(v_c).
  \label{eq:learner-cloning-joint-distribution}
\end{equation}
We note that $\nu_C$ is independent of $L$.

For each color, consider the measure-and-prepare channel
\begin{equation}
  \mathcal M_{c,L}(X)
  :=d_c\int\bra{v_c}X\ket{v_c}
    \left(\ket{v_c}\!\bra{v_c}\right)^{\otimes L}
    \,\mathrm d\mu_c(v_c).
\end{equation}
Theorem~1 of~\cite{Chiribella2011}, applied with one input copy,
$L$ output copies, and local dimension $d_c$, gives
\begin{equation}
  \left\|\mathcal W_{c,L}-\mathcal M_{c,L}\right\|_\diamond
  \le\frac{2d_c}{L+d_c}.
\end{equation}

By the definition of $\nu_C$,
\begin{equation}
  \left(\bigotimes_c\mathcal M_{c,L}\right)(\rhobar_C)
  =\mathbb E_{\nu_C}\left[
    \bigotimes_c\left(\ket{v_c}\!\bra{v_c}\right)^{\otimes L}
  \right].
\end{equation}
Applying both channels to $\bar{\rho}_C$, we obtain
\begin{align*}
  \left\|
    \Omega_{C,L}
    -\mathbb E_{\nu_C}\left[
      \bigotimes_c\left(\ket{v_c}\!\bra{v_c}\right)^{\otimes L}
    \right]
  \right\|_1
  &\le\left\|
    \bigotimes_c\mathcal W_{c,L}
    -\bigotimes_c\mathcal M_{c,L}
  \right\|_\diamond\\
  &\le\sum_c\frac{2d_c}{L+d_c}
   \le\frac{2q}{\theta_0L}.
\end{align*}
The last inequality uses $d_c\le\theta_0^{-1}$ from
Lemma~\ref{lem:learner-spectral-projection}. This proves~(i).

To prove~(ii), for each color $c$, we define the good event POVM operator as 
\begin{equation}
    M_c:=d_c\int\ketbra{v}{v}\cdot \1\lrb{\abs{\braket{\xbar_c}{v}}^2\ge1- \f{\vepclone}{q}}\,\mathrm d\mu_c(v).
\end{equation}
We claim that
\begin{equation}
  M_c\succeq\p{\f{\vepclone}{q}}^{d_c-1}\ket{\xbar_c}\!\bra{\xbar_c}.
  \label{eq:learner-cloning-cap-dominance}
\end{equation}
For $d_c=1$, the claim holds trivially. 
For $d_c\ge2$, we will first show that $\ket{\xbar_c}$ is an eigenvector of $M_c$.
Let $U_c = 2\ketbra{\xbar_c}{\xbar_c} - I_{E_c}$ be a unitary that fixes $\ket{\xbar_c}$ and negates its orthogonal complement.
It is easy to see $ U_cM_cU_c^\dagger = M_c$ and therefore $U_cM_c\ket{\xbar_c} = M_cU_c\ket{\xbar_c}
 = M_c\ket{\xbar_c}$.
Notice that $M_c\ket{\xbar_c}$ is an eigenvector of $U_c$ with eigenvalue $1$, it follows that $M_c\ket{\xbar_c} = \kappa_c\ket{\xbar_c}$ for some scalar $\kappa_c$.

Under $\mu_c$, the variable $t=\abs{\braket{\xbar_c}{v}}^2$
has density $(d_c-1)(1-t)^{d_c-2}$ on $[0,1]$;
see~\cite[Eq.~(4.3)]{ZyczkowskiSommers2005}.
Therefore we have that
\begin{align}
  \bra{\xbar_c}M_c\ket{\xbar_c}
  &=d_c(d_c-1)\int_{1-\vepclone/q}^1 t(1-t)^{d_c-2}\,\mathrm dt\\
  &=\p{\f{\vepclone}{q}}^{d_c-1}\p{d_c-(d_c-1)\f{\vepclone}{q}}
  \ge\p{\f{\vepclone}{q}}^{d_c-1}.
\end{align}
Since $M_c$ is positive semidefinite, this proves
\eqref{eq:learner-cloning-cap-dominance}.

Taking tensor products of these operator inequalities and using
\eqref{eq:learner-cloning-joint-distribution}, we obtain
\begin{align}
  \Pr_{\nu_C}\lrb{\abs{\braket{\xbar_c}{v_c}}^2\ge1- \vepclone/q\text{ for every }c}
  &=\Tr\!\left[\rhobar_C\bigotimes_c M_c\right]\\
  &\ge\p{\f{\vepclone}{q}}^{\sum_c(d_c-1)}
    \bra{\xbar_C}\rhobar_C\ket{\xbar_C}\\
  &\ge\pproj\left(\frac{\vepclone}{q}\right)^{\sum_c(d_c-1)}
   \ge\pclone,
\end{align}
where we use the overlap guarantee from
Lemma~\ref{lem:learner-spectral-projection}.
This proves~(ii).
\end{proof}

\subsubsection{Preparation of the Cloned State $\Omega_{C,L}$}
\label{subsubsec:learner-clone-preparation}
In this section we show in detail how to implement the channels $\mathcal W_{c,L}$ using copies of
$\rho_c$ and the spectral filters from \Cref{sec:blockwise-spectral-projection}.
Readers who are mainly interested in the correctness of the learning algorithm can skip this subsection for forward to the next one.
The difficulty is that we do not know a basis of $E_c$.
We overcome this by preparing the totally antisymmetric state
associated with $E_c$~\cite[Section~2]{JexEtAl2003}: the unique
antisymmetric pure state on $d_c$ registers supported on
$E_c^{\otimes d_c}$.
One register of this state is maximally mixed on $E_c$, so it supplies
the additional register needed by Werner's optimal cloning channel. 

By the definition of Werner's optimal cloning channel, we can append $L-1$ copies of the maximally mixed states for each color block $c$ and $\rhobar_c$, and then project the whole state onto the symmetric subspace. 
The problem of this idea is the success probability will be exponentially small depending on $n$. 
To overcome it and maintain a high probability for the whole cloning channel, 
we add the output registers and implement the projection operations one by one.
And for each step, in order to amplify the success probability for the projection operation, we need the following theorem from \cite{GSLW19}.

\begin{theorem}[Singular vector transformation, rephrased from
{\cite[Theorem~1]{GSLW19}}]
\label{thm:learner-singular-vector-transformation}
Let $U$ be a unitary on a system $S$, and let
$\Pi,\widetilde\Pi$ be projectors.
Write
$\widetilde\Pi U\Pi=\sum_i\sigma_i\ket{u_i}\bra{v_i}$
for a singular value decomposition.
For any $0<\delta\le1$ and $0<\epsilon\le1/2$, define
\begin{equation}
  V:=\sum_{\sigma_i\ge\delta}\ket{u_i}\bra{v_i}.
\end{equation}
There is a unitary circuit $U_{\mathrm{svt}}$ whose induced operation
\begin{equation}
  \mathcal C(X)
  :=
  \operatorname{Tr}_{\mathrm{anc}}\!\left[
    U_{\mathrm{svt}}
    \bigl(X\otimes\ket{0}\!\bra{0}_{\mathrm{anc}}\bigr)
    U_{\mathrm{svt}}^\dagger
  \right]
\end{equation}
satisfies
\begin{equation}
  \left\|
    (\mathcal C\otimes\operatorname{id}_R)(\rho_{SR})
    -
    (V\otimes I_R)\rho_{SR}(V^\dagger\otimes I_R)
  \right\|_1
  \le\epsilon
\end{equation}
for every reference system $R$ and every density operator
$\rho_{SR}$ whose marginal on $S$ is supported on
$\operatorname{span}\{\ket{v_i}:\sigma_i\ge\delta\}$.
The circuit uses
$O(\delta^{-1}\log(1/\epsilon))$ queries to
$U$, $U^\dagger$, and the controlled reflections
$2\Pi-I$ and $2\widetilde\Pi-I$.
\end{theorem}

The following lemma shows how to prepare the cloned state $\Omega_{C,L}$.

\begin{lemma}[Preparation of the cloned state]
\label{lem:learner-clone-preparation}
Assume that all spectral thresholds $\theta_c$ satisfy the conditions of
Lemma~\ref{lem:learner-guarded-filter} with width
$\Delta_{\mathrm{sp}}$, and that $d_c\ge1$ for every color.
For every integer $L\ge1$ and $\xi_{\mathrm{clone}}\in(0,1/2]$, 
we can implement a channel $\widetilde{\mathcal W}_{C,L}$ using copies of $\rho$,
such that, for every state $\sigma$ whose input marginal is supported on $\bigotimes_c E_c$,
\begin{equation}
  \left\|
    (\widetilde{\mathcal W}_{C,L}\otimes\operatorname{id})(\sigma)
    -\left(\bigotimes_c\mathcal W_{c,L}
      \otimes\operatorname{id}\right)(\sigma)
  \right\|_1\le\xi_{\mathrm{clone}}.
  \label{eq:learner-cloner-implementation-error}
\end{equation}
The sample complexity is
\begin{equation}
  \p{\f{1}{\theta_0}}^{O(1/\theta_0)}
  \poly(q,\Delta_{\mathrm{sp}}^{-1})\,
  \widetilde O\!\left(\frac{L^2}{\xi_{\mathrm{clone}}}\right).
  \label{eq:learner-cloner-implementation-cost}
\end{equation}
and the computational complexity is at most an additional factor
$n\lceil\log_2d\rceil+q$ times this bound.
\end{lemma}

\begin{proof}
Fix a nonempty color block $c$. Write $d_c=\dim E_c$.
Recall that $d_c \le 1/\theta_0$ according to \Cref{lem:learner-spectral-projection}.
Let $\Pi_{c,k}^{\mathrm{alt}}$ be the projector onto the antisymmetric
subspace of $\mathcal H_c^{\otimes k}$, 
\begin{equation}
  \Pi_{c,k}^{\mathrm{alt}}
  =\frac1{k!}\sum_{\pi\in\symgroup_k}
    \operatorname{sgn}(\pi)W_{c,\pi},
\end{equation}
where $W_{c,\pi}$ permutes the $k$ copies of $\mathcal H_c$ according to $\pi$.

\paragraph{Finding the dimension and preparing the antisymmetric state.}
Project $k$ independent copies of $\rho_c$ into $E_c$, and then apply
$\Pi_{c,k}^{\mathrm{alt}}$. If $\lambda_{c,1},\ldots,\lambda_{c,d_c}$ are the
eigenvalues selected by $\Pihi_c$, the joint acceptance probability is
\begin{equation}
  \Tr\!\left(\Pi_{c,k}^{\mathrm{alt}}(\Pihi_c\rho_c\Pihi_c)^{\otimes k}\right)
  =\sum_{1\le j_1<\cdots<j_k\le d_c}
    \lambda_{c,j_1}\cdots\lambda_{c,j_k}.
  \label{eq:learner-alternating-rank-test}
\end{equation}
This is zero for $k>d_c$, and at least $\theta_0^k\ge \theta_0^{1/\theta_0} =: p_*$ for
$1\le k\le d_c$. We can therefore find $d_c$ by estimating the probabilities of the antisymmetric subspace projection test for $1\le k\le 1/\theta_0$.
Lemma~\ref{lem:learner-guarded-filter} and a union bound give the
correct dimensions for all colors with failure probability at most $\delta_{\mathrm{dim}}$, and sample complexity
$\theta_0^{-O(1/\theta_0)}\poly(q,\Delta_{\mathrm{sp}}^{-1})\log\p{1/\delta_{\mathrm{dim}}}$.

If any nonempty color block has estimated dimension zero, we return a fixed state on the  output registers and abort. Otherwise, we proceed using the estimated dimensions. When the dimensions of $E_c$'s are positive, this fallback occurs only on the dimension-estimation failure event and will introduce a $2\delta_{\mathrm{dim}}$ $1$-norm error term, which we will discuss later.
On the other hand, if the true dimension of $E_c$ is zero, all eigenvalues of $\rho_c$ lie below $\theta_c$, we will have $\Fprod(\rho) < f_{\min}$ and the learning problem becomes trivial, as follows:
\begin{equation}
  \Fprod(\rho)
  \le \lambda_{\max}(\rho_c)
  < \theta_c
  \le 2\theta_0
  = \frac{f_{\min}\vepproj}{q}
  < f_{\min}.
\end{equation}
At $k=d_c$, the accepted subspace $\wedge^{d_c}E_c$ is one-dimensional.
For any orthonormal basis $\{\ket{e_{c,j}}\}_{j=1}^{d_c}$ of $E_c$, write
\begin{equation}
  \ket{\operatorname{Alt}_{E_c}}
  :=\frac1{\sqrt{d_c!}}\sum_{\pi\in\symgroup_{d_c}}
    \operatorname{sgn}(\pi)
    \ket{e_{c,\pi(1)}\cdots e_{c,\pi(d_c)}},
  \qquad
  \Phi_{E_c}:=\ketbra{\operatorname{Alt}_{E_c}}{\operatorname{Alt}_{E_c}}.
\end{equation}
A change of basis changes this vector only by a global phase.
Moreover,
\begin{equation}
  \label{eq:learner-slater-preparation}
    \Pi_{c,d_c}^{\mathrm{alt}}
    (\Pihi_c\rho_c\Pihi_c)^{\otimes d_c}
    \Pi_{c,d_c}^{\mathrm{alt}}
    =\left(\prod_{j=1}^{d_c}\lambda_{c,j}\right)\Phi_{E_c},
\end{equation}
where $\prod_{j=1}^{d_c}\lambda_{c,j}
    \ge\theta_0^{d_c}\ge p_*$.
Thus the joint projection and antisymmetric test prepares $\Phi_{E_c}$
with probability at least $p_*$. If $\mathsf{P}_c$ denotes one of its
registers and $\mathsf{F}_c$ denotes the remaining $d_c-1$, then
\begin{equation}
  \Tr_{\mathsf{F}_c}(\Phi_{E_c})=\frac{\Pihi_c}{d_c} = \f{I_{E_c}}{d_c}.
  \label{eq:learner-slater-marginal}
\end{equation}
Now we are going to discuss how the error from the high-eigenvalue projector we implement would affect the $\Phi_{E_c}$ we prepare here.
Approximate spectral measurements may leave a small component outside
$E_c$, so we control the error before conditioning on acceptance.
Assume that $d_c$ has been correctly identified and that each spectral
measurement has diamond-norm error at most
$\xi_{\mathrm{proj}}$. Let $\tilde{\Phi}_{E_c}$ be the accepted,
subnormalized output, and set
$p_c:=\prod_{j=1}^{d_c}\lambda_{c,j}\ge p_*$ and
$\widetilde p_c:=\Tr(\tilde{\Phi}_{E_c})$. 
As we invoke the spectral filter $d_c$ times, we have
\begin{equation}
  \|\tilde{\Phi}_{E_c}-p_c\Phi_{E_c}\|_1
  \le d_c\xi_{\mathrm{proj}},
  \qquad
  |\widetilde p_c-p_c|
  \le d_c \xi_{\mathrm{proj}}.
\end{equation}
For $0<\vep_{\mathrm{alt}}\le1$, choosing
$d_c\xi_{\mathrm{proj}}\le p_*\vep_{\mathrm{alt}}/4$
therefore ensures $\widetilde p_c\ge p_*/2$ and
\begin{equation}
  \left\|\frac{\tilde{\Phi}_{E_c}}{\widetilde p_c}
              -\Phi_{E_c}\right\|_1
  \le
  \frac{\|\tilde{\Phi}_{E_c}-p_c\Phi_{E_c}\|_1
        +|\widetilde p_c-p_c|}{\widetilde p_c}
  \le\vep_{\mathrm{alt}}.
\end{equation}
Thus conditioning on acceptance amplifies the accumulated error by at
most a factor $O(p_*^{-1})$. Since $d_c\le\theta_0^{-1}$, it suffices
to implement each spectral filter to error $\xi_{\mathrm{proj}} = O(p_*\vep_{\mathrm{alt}}\theta_0)$. 
We conclude that the normalized state
$\tilde{\Phi}_{E_c}/\widetilde p_c$ differs from $\Phi_{E_c}$ by at most
$\vep_{\mathrm{alt}}$ in trace norm. By \Cref{lem:learner-guarded-filter}, one
attempt uses
$\widetilde O(1/(p_*\Delta_{\mathrm{sp}}^2\theta_0^2\vep_{\mathrm{alt}}))$
copies of $\rho_c$. Since each attempt succeeds with probability at
least $p_*/2$, repeating the preparation uses
\begin{equation}
  \widetilde O\!\left(
    \frac{1}{p_*^2\Delta_{\mathrm{sp}}^2\theta_0^2\vep_{\mathrm{alt}}}
  \right)
  \label{eq:learner-slater-preparation-cost}
\end{equation}
copies in expectation.

\paragraph{Implementing Werner's optimal cloning channel.}
We now explain how to use the antisymmetric state $\Phi_{E_c}$ to implement Werner's optimal cloning channel.
The case $L=1$ is the identity channel, so assume $L\ge2$.
For the moment, consider an exact preparation of $\Phi_{E_c}$.
Suppose the current $\ell$ registers have a state $X_c$ supported on
$\Sym^\ell(E_c)$. Append $\Phi_{E_c}$ and project the current registers
together with $\mathsf P_c$ onto $\Sym^{\ell+1}(E_c)$. By
\eqref{eq:learner-slater-marginal}, tracing out $\mathsf F_c$ gives
the accepted, subnormalized output
\begin{equation}
  \Pi_{\ell+1,E_c}
  \left(X_c\otimes\frac{I_{E_c}}{d_c}\right)
  \Pi_{\ell+1,E_c}.
\end{equation}
Note that $\dim\Sym^\ell(E_c)=\binom{\ell+d_c-1}{\ell}$.
The acceptance probability is
\begin{equation}
  p_{c,\ell}
  =\frac{\dim\Sym^{\ell+1}(E_c)}{d_c\dim\Sym^\ell(E_c)}
  =\frac{\ell+d_c}{d_c(\ell+1)}
  \ge\frac1{d_c}.
  \label{eq:learner-one-clone-success}
\end{equation}
In particular, it is independent of $X_c$. Conditioning on acceptance
therefore implements the channel
\begin{equation}
  \mathcal W_{\ell\to\ell+1}^{E_c}(X_c)
  :=
  \frac{\dim\Sym^{\ell}(E_c)}{\dim\Sym^{\ell+1}(E_c)}
  \Pi_{\ell+1,E_c}(X_c\otimes I_{E_c})\Pi_{\ell+1,E_c}.
  \label{eq:learner-one-clone-channel}
\end{equation}
To avoid multiplying the acceptance probabilities over the $L-1$
additions, we need to keep the $\mathsf F_c$ register and amplify the acceptance probability for each step.
Consider the subspace
$\Sym^\ell(E_c)\otimes E_c^{\otimes d_c}$, and define the projectors
\begin{equation}
  P_{c,\ell}:=I_{\Sym^\ell(E_c)}\otimes\Phi_{E_c},
  \qquad
  Q_{c,\ell}:=\Pi_{\ell+1,E_c}\otimes I_{\mathsf F_c}.
  \label{eq:learner-amplification-projectors}
\end{equation}
By \eqref{eq:learner-one-clone-success}, every unit vector in the
range of $P_{c,\ell}$ is accepted by $Q_{c,\ell}$ with probability
$p_{c,\ell}$. Consequently,
\begin{equation}
  P_{c,\ell}Q_{c,\ell}P_{c,\ell}
  =p_{c,\ell}P_{c,\ell}.
  \label{eq:learner-amplification-scalar}
\end{equation}
Note that $P_{c,\ell}Q_{c,\ell}Q_{c,\ell}P_{c,\ell} = P_{c,\ell}Q_{c,\ell}P_{c,\ell} = p_{c,\ell}P_{c,\ell}$,
it follows that all nonzero singular values of
$Q_{c,\ell}P_{c,\ell}$ equal $\sqrt{p_{c,\ell}}$, and
$Q_{c,\ell}P_{c,\ell}/\sqrt{p_{c,\ell}}$ is an isometry on the
range of $P_{c,\ell}$.

We apply Theorem~\ref{thm:learner-singular-vector-transformation}
with $ U=I, \Pi=P_{c,\ell}, \widetilde\Pi=Q_{c,\ell}$, and therefore
$\widetilde{\Pi} U \Pi = Q_{c,\ell} P_{c,\ell}$;
and $ \delta=\sqrt{p_{c,\ell}}/2,
  \epsilon=\vep_{\mathrm{svt}}. $ Since every nonzero singular value of
$Q_{c,\ell}P_{c,\ell}$ equals $\sqrt{p_{c,\ell}}$,
the theorem applies to the entire range of $P_{c,\ell}$.
Including the appended antisymmetric state, the ideal isometry is
\begin{equation}
  \frac{Q_{c,\ell}P_{c,\ell}}{\sqrt{p_{c,\ell}}}
  \bigl(\ket{\psi_c}\otimes\ket{\operatorname{Alt}_{E_c}}\bigr)
  :=
  \frac{
    (\Pi_{\ell+1,E_c}\otimes I_{\mathsf F_c})
    \bigl(\ket{\psi_c}\otimes
          \ket{\operatorname{Alt}_{E_c}}\bigr)
  }{\sqrt{p_{c,\ell}}},
  \qquad
  \ket{\psi_c}\in\Sym^\ell(E_c).
  \label{eq:learner-one-clone-isometry}
\end{equation}
Discarding $\mathsf F_c$ from its output gives
$\mathcal W_{\ell\to\ell+1}^{E_c}$ by
\eqref{eq:learner-one-clone-channel}.

By \Cref{thm:learner-singular-vector-transformation},
we have a circuit which produces an output state within
$\vep_{\mathrm{svt}}$ in trace distance of the
output of $\mathcal W_{\ell\to\ell+1}^{E_c}$.
This bound holds for the whole state $\bar{\rho}_C$, 
including the other color blocks as the reference system.
Since $U=I$, the queries required by the theorem are
controlled applications of $2P_{c,\ell}-I$ and
$2Q_{c,\ell}-I$, where we will show how to implement them in \Cref{lemma:implementPQ}.
The total number of these operations is
\begin{equation}
  O\!\left(
    \frac{1}{\sqrt{p_{c,\ell}}}
    \log\frac1{\vep_{\mathrm{svt}}}
  \right)  =
  O\!\left(
    \sqrt{d_c}\log\frac1{\vep_{\mathrm{svt}}}
  \right),
  \label{eq:learner-amplified-addition-cost}
\end{equation}
where we used $p_{c,\ell}\ge1/d_c$.

We can further compose them to the desired channel:
\begin{equation}
  \mathcal W_{L-1\to L}^{E_c}\circ\cdots\circ
  \mathcal W_{1\to2}^{E_c}
  =
  \mathcal W_{c,L}.
  \label{eq:learner-sequential-cloner}
\end{equation}
The trace-norm errors therefore sum to at most
$\sum_{\ell=1}^{L-1}\vep_{\mathrm{svt}} = (L-1)\vep_{\mathrm{svt}}$,
by a standard hybrid argument.

\begin{lemma}\label{lemma:implementPQ}
Let $r=d_c\ge1$, $\ell\ge1$, and $m_c=|B_c|\lceil\log_2d\rceil$.
On $\Sym^\ell(E_c)\otimes E_c^{\otimes r}$, the
controlled reflections $2P_{c,\ell}-I$ and $2Q_{c,\ell}-I$ can be
implemented using $r^{O(r)}(m_c+1)$ and
$\widetilde O(\ell(m_c+1))$ elementary gates, respectively,
allowing arbitrary one-qubit gates. Both circuits return their
ancillas to zero and require no copies of $\rho$ or knowledge of a
basis of $E_c$.
\end{lemma}

\begin{proof}
Consider the known Hermitian operators
\begin{equation}
  A_P=I\otimes\Pi_{c,r}^{\mathrm{alt}},
  \qquad
  A_Q=\frac{I+\sum_{j=1}^{\ell}\swap(j,\ell+1)}{\ell+1}
       \otimes I_{\mathsf F_c},
\end{equation}
where $\swap(j,\ell+1)$ swaps the $j$-th and $(\ell+1)$-st block-copy registers, and register $\ell+1$ is $\mathsf P_c$.
Their restrictions to $\Sym^\ell(E_c)\otimes E_c^{\otimes r}$ are $P_{c,\ell}$ and $Q_{c,\ell}$:
the first follows from $\dim\bigwedge^r E_c=1$, and the second from
the coset decomposition of $\symgroup_{\ell+1}$ by $\symgroup_\ell$,
since the first $\ell$ registers are symmetric.

For either operator $A$, prepare a uniform label over its permutation
summands, apply the labelled signed permutation, and unprepare the
label. The resulting unitary $V$ satisfies $J^\dagger VJ=A$, where
$J\ket{\psi}=\ket{0}_{\mathrm{anc}}\ket{\psi}$.
Set $R=V^\dagger \p{2JJ^{\dagger} - I}V$. Then
\begin{equation}
  J^\dagger RJ=2A^2-I.
\end{equation}
On $\Sym^\ell(E_c)\otimes E_c^{\otimes r}$, $A$ is a projector, so
$J^\dagger R J=2A-I$. For every $\ket{\psi}\in\Sym^\ell(E_c)\otimes E_c^{\otimes r}$,
the vector $J^\dagger R J\ket{\psi}$ has norm
$\|\ket{\psi}\|$, because $2A-I$ is unitary on $\Sym^\ell(E_c)\otimes E_c^{\otimes r}$.
Since $R$ and $J$ are isometries, $R J\ket{\psi}$ has the
same norm. It therefore has no component orthogonal to
$\operatorname{ran}J$, and
\begin{equation}
  R J\ket{\psi}=J(2A-I)\ket{\psi}.
\end{equation}
Thus $R$ implements the reflection and returns the
ancillas to zero.
Controlling only the middle phase $2JJ^\dagger - I$ gives its controlled version.
For $A_P$, enumerating $r!$ permutations, each using at most $r$
register swaps, gives the first gate bound. For $A_Q$, selecting
among $\ell+1$ terms costs $O(\ell(m_c+\log(\ell+1)))$ gates,
giving the second bound.
\end{proof}

\paragraph{Overall error and complexity.}
Let $N:=q(L-1)$ and $\tau:=\xi_{\mathrm{clone}}/(8N)$.
Estimate all dimensions with failure probability at most $\delta_{\mathrm{dim}} = \xi_{\mathrm{clone}}/4$,
and choose
$\vep_{\mathrm{alt}}=\vep_{\mathrm{svt}}=\tau$.
For each addition, prepare an independent antisymmetric state $\Phi_{E_c}$
using fresh copies of $\rho$. 
We always try to prepare it with
$O(p_*^{-1}\log(1/\tau))$ attempts and output a fixed state
if all attempts fail. The abort probability is at most $\tau$,
so the state which we prepare differs from $\Phi_{E_c}$
by at most $3\tau$ in trace norm.

Replacing all $\tilde{\Phi}_{E_c}/\tilde{p}_c$ we prepared by the ideal $\Phi_{E_c}$
therefore changes the final output by at most $3N\tau$. 
The amplification errors contribute at most
$N\vep_{\mathrm{svt}} = N\tau$.
The total trace-norm error is at most
\begin{equation}
  2\cdot\frac{\xi_{\mathrm{clone}}}{4}+3N\tau+N\tau
  =\xi_{\mathrm{clone}},
\end{equation}
proving \eqref{eq:learner-cloner-implementation-error}.

By \eqref{eq:learner-slater-preparation-cost}, the total number
of copies, including dimension estimation, is
\begin{equation}
  \widetilde O\!\left(
    \frac{q(L-1)}
    {p_*^2\Delta_{\mathrm{sp}}^2\theta_0^2\vep_{\mathrm{alt}}}
  \right)
  +
  \theta_0^{-O(1/\theta_0)}
  \poly(q,\Delta_{\mathrm{sp}}^{-1})\log(1/\delta_{\mathrm{dim}})
  =
  \left(\frac1{\theta_0}\right)^{O(1/\theta_0)}
  \poly(q,\Delta_{\mathrm{sp}}^{-1})
  \widetilde O\!\left(\frac{L^2}{\xi_{\mathrm{clone}}}\right).
\end{equation}
Singular vector transformation uses only known permutations and phase operations,
and hence requires no additional copies.
We also need to implement the controlled permutation operators with elementary gates, combining with the spectral filter gate bounds from \Cref{lem:learner-guarded-filter}, we have another multiplicative factor $n\lceil\log_2d\rceil+q$ for the computational complexity.
\end{proof}

\begin{parameterbox}
Now we are going to explicitly state our choice of parameters which appear in this section.
These parameters will lead to the desired learning result, \Cref{thm:ProductLearner}.
Recall $\vep_0$, $q$, $\theta_0$, and $\pproj$ from the preceding
subsections. Set
\begin{equation}
  \begin{aligned}
    \vepclone&:=\frac{\vep_0}{10^4},\\
    \pclone&=\pproj\left(\frac{\vepclone}{q}\right)^{
      q(\lceil\theta_0^{-1}\rceil-1)},\\
    \xi_{\mathrm{clone}}&:=\frac{\pclone}{100}.
  \end{aligned}
\end{equation}
By Lemma~\ref{lem:learner-werner-cloning}, the good-vector event has
probability at least $\pclone$. Choose
\begin{equation}
  L\ge\max\left\{2,
    \left\lceil\frac{100q}{\theta_0\pclone}\right\rceil\right\}.
\end{equation}
Then the measure-and-prepare error in
\eqref{eq:learner-cloning-comparison} is at most $\pclone/50$, while
Lemma~\ref{lem:learner-clone-preparation} implements the cloning channel up to
trace-norm error at most $\pclone/100$. The accepted projection from the
preceding subsection has error at most
$\xi_{\mathrm{clone}}=\pclone/100$, so all three errors sum to at most
$\pclone/25$. We may increase
$L$ in the final step to supply enough copies to the high-fidelity
learner. Finally, the previous choices
$q=\Theta(\vep_0^{-3})$, $\theta_0=\Theta(\vep_0^5)$, and
$\pproj=\Theta(\vep_0)$ give
\begin{equation}
  \log\frac1{\pclone}
  =O\!\left(\vep_0^{-8}\log\frac1{\vep_0}\right).
\end{equation}
\end{parameterbox}
\subsection{The FPT Learning Algorithm}
\label{subsec:learner-complete}

With the subroutines above, we are now ready to state the complete learning algorithm.

\begin{protocol}[FPT closest product state learner]
\label{prot:fpt-product-learner}
Given $\vep,\delta$ and copies of $\rho$, proceed as follows.
\begin{enumerate}
  \item Use the parameters of the preceding subsections. Choose the
  clone count $L$ to satisfy
  $L\ge\max\{2,\lceil100q/(\theta_0\pclone)\rceil\}$ and to supply
  enough copies for the high-fidelity learner on every nonempty color block
  with accuracy $\vep_{\mathrm{hfl}}=\vepclone/q$ and failure probability $1/(100q)$.
  Set the projection-attempt repetition count $T$ and overall trial count $R$ to
  \begin{equation}
    T:=\left\lceil\frac{2}{\pproj}\log 10\right\rceil,
    \quad
    R:=\left\lceil\frac{2}{\pclone}\log\frac4\delta\right\rceil.
  \end{equation}
  \item In each of $R$ independent trials, draw a coloring and spectral
  thresholds as in Lemma~\ref{lem:learner-random-guard}.
  Try the blockwise spectral projection on fresh copies of
  $\rho$ at most $T$ times, and discard the trial if none accepts.
  On acceptance, apply the implemented $1\to L$ cloning channel, feed its $L$ physical
  registers for each nonempty color to the high-fidelity learner and
  record the tensor product $\ket z$ of the resulting product states.
  \item Add any fixed pure product state to the candidate list.
  Using fresh copies of $\rho$, estimate each candidate's fidelity
  to accuracy $\vep_0/16$ with failure probability at most
  $\delta/[2(R+1)]$. Return the candidate with the largest estimate.
\end{enumerate}
\end{protocol}

The next lemma relates the local learners' outputs for each color to the original state $\rho$. 

\begin{lemma}[The local learner outputs the closest product state]
\label{lem:learner-output-quality}
Suppose $F_C(\rho)\ge f_{\min}$, 
fix $\ket{x_C} = \bigotimes_{c\in[q]}\ket{x_c}$ achieving
$F_C(\rho)$ and recall that $f_c= \Fprod(x_c)$.
Assume that the coloring satisfies
$\sum_c(1-f_c^2)\le\vepcolor$ and that the cloned-state sample
$\bigotimes_{c\in[q]}\ket{v_c}$ satisfies
$\sum_c(1-|\langle\xbar_c|v_c\rangle|^2)\le\vepclone$.
Apply the high-fidelity learner to $L$ copies of each $\ket{v_c}$
with accuracy $\vepclone/q$ and failure probability
$1/(100q)$. With probability at least $99/100$, its product output
$\ket z=\bigotimes_c\ket{z_c}$ satisfies
\begin{equation}
  \bra z\rho\ket z
  \ge \Fprod(\rho)-\f{\vep_0}{50}
  \label{eq:learner-output-quality}
\end{equation}
\end{lemma}

\begin{proof}
For each color, choose a product state $\ket{u_c}$ such that 
$|\langle u_c|x_c\rangle|^2=f_c$. The pure-state distance
$D(y,y'):=\sqrt{1-|\langle y|y'\rangle|^2}$ satisfies the triangle
inequality. Applying it along
$u_c,x_c,\xbar_c,v_c$ gives
\begin{align}
1-\Fprod(v_c)
&\le
\Bigl(
\sqrt{1-f_c}
+\sqrt{1-|\langle x_c|\xbar_c\rangle|^2}
+\sqrt{1-|\langle\xbar_c|v_c\rangle|^2}
\Bigr)^2 \\
&\le 3\Bigl[
(1-f_c)
+(1-|\langle x_c|\xbar_c\rangle|^2)
+(1-|\langle\xbar_c|v_c\rangle|^2)
\Bigr].\label{eq:learner:overall-ineq}
\end{align}
The coloring assumption gives
$\sum_c(1-f_c)\le\vepcolor$, and
Lemma~\ref{lem:learner-spectral-projection} gives
$\sum_c(1-|\langle x_c|\xbar_c\rangle|^2)\le\vepproj$.
Thus, for every $c$, the high-fidelity learner's promise holds at accuracy $\vep_{\mathrm{hfl}}=\vepclone/q$,
\begin{equation}
  1-\Fprod(v_c)
  \le 3(\vepcolor+\vepproj+\vepclone)
  <\frac16-\frac{\vepclone}{q}.
\end{equation}
Therefore, with probability at least
$1-q/(100q)=99/100$, all the outputs from the high-fidelity closest product state learner satisfy
$|\langle z_c|v_c\rangle|^2
 \ge\Fprod(v_c)-\vepclone/q$.
Applying the same triangle inequality along
$z_c,v_c,\xbar_c,x_c$ gives
\begin{align}
1-|\langle z_c|x_c\rangle|^2
&\le 3\Bigl[
    (1-|\langle z_c|v_c\rangle|^2)
   +(1-|\langle v_c|\xbar_c\rangle|^2)
   +(1-|\langle\xbar_c|x_c\rangle|^2)
  \Bigr] \\
&\le 3\Bigl[
    (1-\Fprod(v_c))+\frac{\vepclone}{q}
   +(1-|\langle v_c|\xbar_c\rangle|^2)
   +(1-|\langle\xbar_c|x_c\rangle|^2)
  \Bigr] \\
&\le 12\Bigl[
    (1-|\langle v_c|\xbar_c\rangle|^2)
   +(1-|\langle\xbar_c|x_c\rangle|^2)
   +(1-f_c)
  \Bigr]
  +\frac{3\vepclone}{q},
\end{align}
where for the third inequality we apply \Cref{eq:learner:overall-ineq}.
Summing over colors yields
\begin{equation}
  e_{\mathrm{out}}:=\sum_c(1-|\langle z_c|x_c\rangle|^2)
  \le12(\vepcolor+\vepproj+\vepclone)+3\vepclone
  =\frac{39\vep_0}{10^4}.
  \label{eq:learner-output-color-defect}
\end{equation}
We now use optimality of $\ket{x_C}$ to convert the sum of infidelity $e_{\mathrm{out}}=O(\vep)$ into an $O(\vep)$ loss in fidelity with $\rho$.
Write $\ket z=\alpha\ket{x_C}+\ket{w_1}+\ket{w_{\ge2}}$,
where $\ket{w_1}$ and $\ket{w_{\ge2}}$ contain terms in which,
respectively, one or at least two color blocks are orthogonal to
the corresponding $\ket{x_c}$. Put
$e_{\mathrm{out}, c}:=1-|\langle z_c|x_c\rangle|^2$. Measuring each color with
$\{\ket{x_c}\!\bra{x_c},I-\ket{x_c}\!\bra{x_c}\}$ gives independent
binary outcomes, with deviation probability $e_{\mathrm{out},c}$. Thus
$|\alpha|^2=\prod_c(1-e_{\mathrm{out},c})$ and 
\begin{equation}
1-|\alpha|^2\le e_{\mathrm{out}}.
\end{equation}
The probability of at least two deviations is
$\|w_{\ge2}\|_2^2$. Corollary~4.3 of
\cite{bakshi2024learningclosestproductstate} gives
\begin{equation}
  \|w_{\ge2}\|_2^2\le 2(1-|\alpha|)^2.
\end{equation}
Since $2|\alpha|\le1+|\alpha|$, it follows that
\begin{equation}
  2|\alpha|\|w_{\ge2}\|_2
  \le 2\sqrt2|\alpha|(1-|\alpha|)
  \le\sqrt2(1-|\alpha|^2)
  \le\sqrt2 e_{\mathrm{out}}.
\end{equation}
Because $\ket{x_C}$ maximizes $F_C(\rho)$, each $\ket{x_c}$ is an
eigenvector of its conditional operator, so
$\bra{x_C}\rho\ket{w_1}=0$. Positivity of $\rho$ gives
$\bra{w_1+w_{\ge2}}\rho\ket{w_1+w_{\ge2}}\ge0$; together with
$\|\rho\|_\infty\le1$ and $F_C(\rho)\le1$, this yields
\begin{align}
  \bra z\rho\ket z
  &\ge|\alpha|^2\bra{x_C}\rho\ket{x_C}
    +2\operatorname{Re}\!\left(\overline{\alpha}\,
      \bra{x_C}\rho\ket{w_{\ge2}}\right)\\
  &\ge |\alpha|^2\bra{x_C}\rho\ket{x_C}
    -2|\alpha|\|w_{\ge2}\|_2\\
  &\ge F_C(\rho)-(1+\sqrt2)e_{\mathrm{out}}.
\end{align}
Finally, $F_C(\rho)\ge\Fprod(\rho)$ and
$(1+\sqrt2)e_{\mathrm{out}}<\vep_0/50$, proving
\eqref{eq:learner-output-quality}.
\end{proof}

Now we are ready to prove our main theorem for learning the closest product state.

\ProductLearner*
\begin{proof}
Run Protocol~\ref{prot:fpt-product-learner}. If
$\Fprod(\rho)<f_{\min}=\vep_0/8$, the fixed product-state candidate
already meets the claimed bound. Otherwise, every coloring has
$F_C(\rho)\ge f_{\min}$. A coloring is good with probability at least
$3/4$, and all thresholds $\theta_c$'s simultaneously satisfy the gap condition with probability at least $7/8$;
both hold with probability at least $5/8$.

On this event, each projection attempt succeeds with probability at
least $\pproj/2$, so the chance that all $T$ attempts fail is at most
$1/10$. For an accepted ideal sample, the good-vector event of
Lemma~\ref{lem:learner-werner-cloning} has probability at least
$\pclone$, and Lemma~\ref{lem:learner-output-quality} then succeeds
with probability at least $99/100$. The total trace-norm error between
the implemented and ideal accepted outputs is at most $\pclone/25$.
Consequently, a trial records a candidate of fidelity at least
$\Fprod(\rho)-\vep_0/50$ with probability at least
\begin{equation}
  \frac58\cdot\frac9{10}
  \left(\frac{99\pclone}{100}-\frac{\pclone}{25}\right)
  \ge\frac{\pclone}{2}.
\end{equation}
By the choice of $R$, the probability that no trial records such a
candidate is at most $\delta/4$.

With probability at least $1-\delta/2$, all fidelity estimates are
within $\vep_0/16$ of their true values. On these two events, choosing
the largest estimate gives
\begin{equation}
  \bra{\wu}\rho\ket{\wu}
  \ge\Fprod(\rho)-\frac{\vep_0}{50}-\frac{\vep_0}{8}
  \ge\Fprod(\rho)-\vep.
\end{equation}
The total failure probability is less than $\delta$. We can take
$L=\widetilde O(q/(\theta_0\pclone)+ndq/\vepclone)$.
Finally,
$q=\Theta(\vep_0^{-3})$, $\theta_0=\Theta(\vep_0^5)$, and
$\log(1/\pclone)=O(\vep_0^{-8}\log(1/\vep_0))$.
The cloning channel's quadratic dependence on $L$ gives sample complexity
$2^{\widetilde O(\vep^{-8})}(nd)^2
\poly(\log(nd/\delta))$.
Its additional gate factor $n\lceil\log_2 d\rceil+q$,
together with the projection, high-fidelity learner, and
fidelity-estimation costs, gives computational complexity
$2^{\widetilde O(\vep^{-8})}n^3d^2
\poly(\log(nd/\delta))$.
\end{proof}

\bibliography{reference}

@misc{grewal2024improvedstabilizerestimationbell,
      title={Improved Stabilizer Estimation via Bell Difference Sampling}, 
      author={Sabee Grewal and Vishnu Iyer and William Kretschmer and Daniel Liang},
      year={2024},
      eprint={2304.13915},
      archivePrefix={arXiv},
      primaryClass={quant-ph},
      url={https://arxiv.org/abs/2304.13915}, 
}

@inproceedings{brandao2013product,
  title={Product-state approximations to quantum ground states},
  author={Brandao, Fernando GSL and Harrow, Aram W},
  booktitle={Proceedings of the forty-fifth annual ACM symposium on Theory of computing},
  pages={871--880},
  year={2013}
}

@inproceedings{fischer2005tolerant,
  title={Tolerant versus intolerant testing for boolean properties},
  author={Fischer, Eldar and Fortnow, Lance},
  booktitle={20th Annual IEEE Conference on Computational Complexity (CCC'05)},
  pages={135--140},
  year={2005},
  organization={IEEE}
}

@article{cramer2010efficient,
  title={Efficient quantum state tomography},
  author={Cramer, Marcus and Plenio, Martin B and Flammia, Steven T and Somma, Rolando and Gross, David and Bartlett, Stephen D and Landon-Cardinal, Olivier and Poulin, David and Liu, Yi-Kai},
  journal={Nature communications},
  volume={1},
  number={1},
  pages={149},
  year={2010},
  publisher={Nature Publishing Group UK London}
}

@misc{bao2025toleranttestingstabilizerstates,
      title={Tolerant testing of stabilizer states with a polynomial gap via a generalized uncertainty relation}, 
      author={Zongbo Bao and Philippe van Dordrecht and Jonas Helsen},
      year={2025},
      eprint={2410.21811},
      archivePrefix={arXiv},
      primaryClass={quant-ph},
      url={https://arxiv.org/abs/2410.21811}, 
}

@book{gs009,
 author = {Canonne, Cl{\'{e}}ment L.},
 title = {A Survey on Distribution Testing: Your Data is Big. But is it Blue?},
 year = {2020},
 pages = {1--100},
 doi = {10.4086/toc.gs.2020.009},
 publisher = {Theory of Computing Library},
 number = {9},
 series = {Graduate Surveys},
 URL = {http://www.theoryofcomputing.org/library.html},
}

@article{Gross_2021,
  title   = {Schur--Weyl Duality for the Clifford Group with Applications:
             Property Testing, a Robust Hudson Theorem, and de Finetti Representations},
  author  = {Gross, David and Nezami, Sepehr and Walter, Michael},
  journal = {Communications in Mathematical Physics},
  volume  = {385},
  pages   = {1325--1393},
  year    = {2021},
  doi     = {10.1007/s00220-021-04118-7}
}

@misc{arunachalam2026optimalstabilizertestinglearning,
  title         = {Optimal Stabilizer Testing and Learning with Limited Quantum Memory},
  author        = {Arunachalam, Srinivasan and Schatzki, Louis},
  year          = {2026},
  eprint        = {2607.02444},
  archivePrefix = {arXiv},
  primaryClass  = {quant-ph}
}

@article{montanaro2013survey,
  title={A survey of quantum property testing},
  author={Montanaro, Ashley and De Wolf, Ronald},
  journal={arXiv preprint arXiv:1310.2035},
  year={2013}
}

@article{harrow2013testing,
  title={Testing product states, quantum Merlin-Arthur games and tensor optimization},
  author={Harrow, Aram W and Montanaro, Ashley},
  journal={Journal of the ACM (JACM)},
  volume={60},
  number={1},
  pages={1--43},
  year={2013},
  publisher={ACM New York, NY, USA}
}

@misc{beckey2026optimalanalysisproducttest,
      title={An Optimal Analysis of the Product Test}, 
      author={Jacob Beckey and Fernando Granha Jeronimo and Pei Wu},
      year={2026},
      eprint={2607.21477},
      archivePrefix={arXiv},
      primaryClass={quant-ph},
      url={https://arxiv.org/abs/2607.21477}, 
}

@misc{she2022unitarypropertytestinglower,
      title={Unitary property testing lower bounds by polynomials}, 
      author={Adrian She and Henry Yuen},
      year={2022},
      eprint={2210.05885},
      archivePrefix={arXiv},
      primaryClass={quant-ph},
      url={https://arxiv.org/abs/2210.05885}, 
}

@inproceedings{soleimanifar2022testing,
  title={Testing matrix product states},
  author={Soleimanifar, Mehdi and Wright, John},
  booktitle={Proceedings of the 2022 Annual ACM-SIAM Symposium on Discrete Algorithms (SODA)},
  pages={1679--1701},
  year={2022},
  organization={SIAM}
}

@misc{bakshi2024learningclosestproductstate,
      title={Learning the closest product state}, 
      author={Ainesh Bakshi and John Bostanci and William Kretschmer and Zeph Landau and Jerry Li and Allen Liu and Ryan O'Donnell and Ewin Tang},
      year={2024},
      eprint={2411.04283},
      archivePrefix={arXiv},
      primaryClass={quant-ph},
      url={https://arxiv.org/abs/2411.04283}, 
}

@misc{arunachalam2024polynomialtimetoleranttestingstabilizer,
      title={Polynomial-time tolerant testing stabilizer states}, 
      author={Srinivasan Arunachalam and Arkopal Dutt},
      year={2024},
      eprint={2408.06289},
      archivePrefix={arXiv},
      primaryClass={quant-ph},
      url={https://arxiv.org/abs/2408.06289}, 
}

@misc{canonne2021pricetolerancedistributiontesting,
      title={The Price of Tolerance in Distribution Testing}, 
      author={Clément L. Canonne and Ayush Jain and Gautam Kamath and Jerry Li},
      year={2021},
      eprint={2106.13414},
      archivePrefix={arXiv},
      primaryClass={cs.DS},
      url={https://arxiv.org/abs/2106.13414}, 
}

@article{slade2023hypercube,
  author  = {Slade, Gordon},
  title   = {Self-avoiding walk on the hypercube},
  journal = {Random Structures \& Algorithms},
  volume  = {62},
  number  = {3},
  pages   = {689--736},
  year    = {2023},
  doi     = {10.1002/rsa.21117}
}

@inproceedings{GSLW19,
  author = {Gily{\'e}n, Andr{\'a}s and Su, Yuan and Low, Guang Hao
            and Wiebe, Nathan},
  title = {Quantum singular value transformation and beyond:
           exponential improvements for quantum matrix arithmetics},
  booktitle = {Proceedings of the 51st Annual ACM SIGACT Symposium
               on Theory of Computing},
  series = {STOC '19},
  year = {2019},
  pages = {193--204},
  publisher = {Association for Computing Machinery},
  doi = {10.1145/3313276.3316366},
  url = {https://doi.org/10.1145/3313276.3316366}
}

@article{WangZhang2025Samplizer,
  author = {Wang, Qisheng and Zhang, Zhicheng},
  title = {Time-Efficient Quantum Entropy Estimator via Samplizer},
  journal = {IEEE Transactions on Information Theory},
  volume = {71},
  number = {12},
  pages = {9569--9599},
  year = {2025},
  doi = {10.1109/TIT.2025.3576137},
  url = {https://doi.org/10.1109/TIT.2025.3576137}
}

@article{Rall2021BlockMeasurement,
  author = {Rall, Patrick},
  title = {Faster Coherent Quantum Algorithms for Phase, Energy,
           and Amplitude Estimation},
  journal = {Quantum},
  volume = {5},
  pages = {566},
  year = {2021},
  doi = {10.22331/q-2021-10-19-566},
  url = {https://doi.org/10.22331/q-2021-10-19-566}
}

@article{Werner1998,
  author = {Werner, R. F.},
  title = {Optimal cloning of pure states},
  journal = {Physical Review A},
  volume = {58},
  number = {3},
  pages = {1827--1832},
  year = {1998},
  doi = {10.1103/PhysRevA.58.1827},
  eprint = {quant-ph/9804001},
  archivePrefix = {arXiv},
  url = {https://arxiv.org/abs/quant-ph/9804001}
}

@inproceedings{Chiribella2011,
  author = {Chiribella, Giulio},
  title = {On quantum estimation, quantum cloning and finite quantum
           de {Finetti} theorems},
  booktitle = {Theory of Quantum Computation, Communication, and Cryptography},
  series = {Lecture Notes in Computer Science},
  volume = {6519},
  pages = {9--25},
  publisher = {Springer},
  year = {2011},
  doi = {10.1007/978-3-642-18073-6_2},
  eprint = {1010.1875},
  archivePrefix = {arXiv},
  primaryClass = {quant-ph},
  url = {https://arxiv.org/abs/1010.1875}
}

@article{ZyczkowskiSommers2005,
  author = {{\.{Z}}yczkowski, Karol and Sommers, Hans-J{\"u}rgen},
  title = {Average fidelity between random quantum states},
  journal = {Physical Review A},
  volume = {71},
  number = {3},
  pages = {032313},
  year = {2005},
  doi = {10.1103/PhysRevA.71.032313},
  eprint = {quant-ph/0311117},
  archivePrefix = {arXiv},
  url = {https://arxiv.org/abs/quant-ph/0311117}
}

@article{JexEtAl2003,
  author = {Jex, I. and Alber, G. and Barnett, S. M. and Delgado, A.},
  title = {Antisymmetric multi-partite quantum states and their applications},
  journal = {Fortschritte der Physik},
  volume = {51},
  number = {2--3},
  pages = {172--178},
  year = {2003},
  doi = {10.1002/prop.200310021},
  eprint = {quant-ph/0209062},
  archivePrefix = {arXiv},
  url = {https://arxiv.org/abs/quant-ph/0209062}
}

@misc{weinbrenner2026complete,
  author = {Weinbrenner, Lisa T. and Rico, Albert and Goodenough, Kenneth and Yu, Xiao-Dong and G{\"u}hne, Otfried},
  title = {Complete Hierarchies for the Geometric Measure of Entanglement},
  year = {2026},
  eprint = {2601.23243},
  archivePrefix = {arXiv},
  primaryClass = {quant-ph},
  doi = {10.48550/arXiv.2601.23243},
  url = {https://arxiv.org/abs/2601.23243}
}

@misc{issa2021applications,
  author = {Issa, Hassan and Mourad, Bassam and Abbas, Hassan},
  title = {New Applications to Combinatorics and Invariant Matrix Norms of an Integral Representation of Natural Powers of the Numerical Values},
  year = {2021},
  eprint = {2106.03810},
  archivePrefix = {arXiv},
  primaryClass = {math.FA},
  doi = {10.48550/arXiv.2106.03810},
  url = {https://arxiv.org/abs/2106.03810}
}

@misc{burchardt2025highdimensional,
  author = {Burchardt, Adam and Fei, Jiani and Grinko, Dmitry and Larocca, Martin and Ozols, Maris and Timmerman, Sydney and Visnevskyi, Vladyslav},
  title = {High-dimensional quantum Schur transforms},
  year = {2025},
  eprint = {2509.22640},
  archivePrefix = {arXiv},
  primaryClass = {quant-ph},
  doi = {10.48550/arXiv.2509.22640},
  url = {https://arxiv.org/abs/2509.22640}
}

@article{adiprasito2020theorems,
  title={Theorems of Carath{\'e}odory, Helly, and Tverberg without dimension},
  author={Adiprasito, Karim and B{\'a}r{\'a}ny, Imre and Mustafa, Nabil H and Terpai, Tam{\'a}s},
  journal={Discrete \& Computational Geometry},
  volume={64},
  number={2},
  pages={233--258},
  year={2020},
  publisher={Springer}
}

@misc{mehraban2024improvedboundstestinglow,
      title={Improved bounds for testing low stabilizer complexity states}, 
      author={Saeed Mehraban and Mehrdad Tahmasbi},
      year={2024},
      eprint={2410.24202},
      archivePrefix={arXiv},
      primaryClass={quant-ph},
      url={https://arxiv.org/abs/2410.24202}, 
}
\bibliographystyle{halpha-abbrv}

\clearpage
\appendix

\section{Proof of \Cref{thm:hfl-qudit}}\label{app:proof-of-[thm:hfl-qudit]}
In this section we give a proof of the qudit extension of the high fidelity product learner of \cite{bakshi2024learningclosestproductstate}. This appendix is primarily written by GPT6.0 Astra, with the authors checking correctness and making some relatively minor edits for reability. The proof itself is a straightforward extension of the argument in \cite[Section $4$]{bakshi2024learningclosestproductstate}. However some care must be taken to make the dependence on $d$ explicit, as the complexity of the main text learning algorithm directly inherits this complexity.

\HFL*

\begin{proof}
We adapt the high-fidelity learner of
\cite[Section~4]{bakshi2024learningclosestproductstate}, using the
qudit parametrization from its Remark~3.5. We explain the dependence
on $d$; the convergence and recursion arguments are unchanged.

\paragraph{Local optimization on qudits.}
For a block of $m$ qudits, work in local bases where the current
product-state candidate is $\ket{0^m}$. Put $D=m(d-1)$ and define
\begin{equation}
  z_{i,j}:=\bra{e_{i,j}}\rho\ket{0^m},
  \qquad
  \ket{e_{i,j}}:=\ket{0}^{\otimes(i-1)}\ket{j}
                  \ket{0}^{\otimes(m-i)},
  \quad i\in[m],\ j\in[d-1].
\end{equation}
The product update associated with a vector $a\in\C^D$ is
\begin{equation}
  \ket{\pi_{a/10}}
  :=\bigotimes_{i=1}^m
  \frac{\ket{0}+\sum_{j=1}^{d-1}(a_{i,j}/10)\ket{j}}
       {\sqrt{1+\|a_i\|_2^2/100}}.
\end{equation}
In the proofs of Corollary~4.7 and Theorem~4.9 of
\cite{bakshi2024learningclosestproductstate}, count a site as excited
when it is orthogonal to $\ket{0}$. These excitation indicators still
have a product Bernoulli distribution for a product vector, so the
weight-zero, weight-one, and higher-weight estimates are unchanged.
Replacing scalar local coordinates by vectors and scalar products by
inner products therefore gives the same optimality-gap and update
inequalities, with the same constants.

To estimate $z$, adapt Lemma~4.10 of
\cite{bakshi2024learningclosestproductstate}. Coherently transfer
$\ket{0^m}$ and the $D$ vectors $\ket{e_{i,j}}$ to distinct labels
in a register of dimension $D+2$, resetting the original registers
on this subspace; send its orthogonal complement to a separate
flag label with the remaining information in the original registers.
After discarding the original registers, all matrix elements
$z_{i,j}$ are preserved. Binary padding and Clifford shadows on the
label register estimate the $2D$ real and imaginary observables,
each with variance at most $3/2$. The vector median-of-means argument
of that lemma thus gives $\|a-z\|_2\le\eta$ with probability
$1-\beta$, using $O(D\eta^{-2}\log(1/\beta))$ copies.
Local basis changes use $\widetilde O(d)$ gates per site by
state preparation and its inverse; label transfer and shadow
processing cost $\widetilde O(md)$ per sample. The total work is
$\widetilde O((md)^2\eta^{-2})$.

Algorithm~4.11 and Corollary~4.17 of
\cite{bakshi2024learningclosestproductstate} now apply with their
gap parameter $C=1/6$: a seed of fidelity at least $2/3$ and
$\Fprod(\rho)\ge5/6$ yield loss at most $\tau$, with
$\widetilde O(md/\tau)$ copies and
$\widetilde O((md)^2/\tau)$ work.
Cap the estimator calls at the bounds in their Lemma~4.15 and return
a fixed product state if a cap is exceeded. On the event that all
estimates are accurate, these caps are never exceeded; hence this
also gives bounded resource use on unsuccessful runs.

\paragraph{The one-qudit seed.}
The one-qubit base case of
\cite[Algorithm~4.19]{bakshi2024learningclosestproductstate}
uses constant-size tomography. For one qudit, embed $\C^d$ in the
next power-of-two dimension, which is less than $2d$, and average
global-Clifford shadows $\widehat\rho$ of the embedded state, identifying
$\rho$ with its zero-padded embedding.
The shadow identity used in Bakshi et al.'s Lemma~4.10 gives
\begin{equation}
  \mathbb E\widehat\rho=\rho,
  \qquad
  \left\|\mathbb E(\widehat\rho-\rho)^2\right\|_\infty<4d.
\end{equation}
Matrix concentration gives operator-norm error at most $1/48$
from $O(d\log(d/\delta))$ copies. Project the estimate back to the
physical $\C^d$ and find a leading eigenvector to Rayleigh error
$1/48$. Since $\lambda_{\max}(\rho)>5/6$, its true value is at least
$\lambda_{\max}(\rho)-3/48>2/3$. Keeping the shadows as rank-one
terms and using the constant top-eigenvalue gap makes the eigensolve
cost $\widetilde O(d^2\log(1/\delta))$. The local optimizer then
achieves loss $\tau$, so the one-qudit call has the same
$\widetilde O(d/\tau)$ copy and $\widetilde O(d^2/\tau)$ work bounds.

\paragraph{Divide and conquer.}
Run the construction with internal accuracy $\tau=\vep_{\mathrm{hfl}}/2$; this
handles the non-strict promise in \Cref{thm:hfl-qudit}. Split the
current sites into balanced left and right halves, learn the two
marginals recursively, tensor their outputs, and apply the local
optimizer. Restricting a global product optimizer to a marginal
shows that every recursive call retains product fidelity at least
$5/6+\vep_{\mathrm{hfl}}>5/6+\tau$. Each successful child therefore has value
above $5/6$. The projector union bound of
\cite[Lemma~4.18]{bakshi2024learningclosestproductstate}
makes their tensor product a seed of value at least $2/3$.
At each node give half its failure budget to local optimization and
one quarter to each child; at a leaf, split the budget between the
seed routine and optimization. Induction gives success probability
at least $1-\delta$ and final loss at most $\tau\le\vep_{\mathrm{hfl}}$.
This is the recursion of
\cite[Algorithm~4.19 and Theorem~4.20]{bakshi2024learningclosestproductstate}.

At recursion depth $\ell$, the largest block has at most
$\lceil n/2^\ell\rceil$ sites. One fresh global copy supplies marginal
samples to all disjoint blocks at that depth, so the copy demands
take their maximum. Summing over depths gives
$\widetilde O(nd/\vep_{\mathrm{hfl}})$ copies. Gates and classical work add across
blocks; the sum of squared block sizes over all depths is $O(n^2)$,
giving $\widetilde O((nd)^2/\vep_{\mathrm{hfl}})$ work. Finite-precision operations
add only logarithmic factors after reserving constant fractions of
the accuracy and failure budgets for rounding and gate synthesis.
This proves \Cref{thm:hfl-qudit}.
\end{proof}

\end{document}